%% file: main.tex
\documentclass[11pt]{article}
\usepackage{graphicx} 
\usepackage{amsmath,amssymb,amsthm}
\usepackage{bm,fullpage,hyperref}
\usepackage{color}
\usepackage{booktabs,multirow}
\usepackage{comment}

\usepackage{algpseudocode}
\usepackage[linesnumbered,ruled,vlined]{algorithm2e}
\SetKwProg{Procedure}{Procedure}{}{}

\allowdisplaybreaks

\newcommand{\E}{\mathop{\mathbf{E}}}
\newcommand{\EMW}{\mathrm{EM}_{\mathrm{w}}}

\newcommand{\TV}{\mathrm{TV}}
\newcommand{\floor}[1]{\left \lfloor{#1}\right \rfloor}
\newcommand{\ceil}[1]{\left \lceil{#1}\right \rceil}

\newcommand{\OPT}{\mathrm{OPT}}
\newcommand{\OPTLP}{\mathrm{OPT}_{\mathrm{LP}}}
\newtheorem{theorem}{Theorem}[section]
\newtheorem{lemma}[theorem]{Lemma}

\newtheorem{corollary}[theorem]{Corollary}

\theoremstyle{definition}
\newtheorem{Definition}[theorem]{Definition}

\title{Lipschitz Continuous Algorithms for Covering Problems}
\author{Soh Kumabe \and Yuichi Yoshida}

\begin{document}

\maketitle

\begin{abstract}
    Combinatorial algorithms are widely used for decision-making and knowledge discovery, and it is important to ensure that their output remains stable even when subjected to small perturbations in the input.
    Failure to do so can lead to several problems, including costly decisions, reduced user trust, potential security concerns, and lack of replicability.
    Unfortunately, many fundamental combinatorial algorithms are vulnerable to small input perturbations.
    To address the impact of input perturbations on algorithms for weighted graph problems, Kumabe and Yoshida (FOCS'23) recently introduced the concept of Lipschitz continuity of algorithms.
    This work explores this approach and designs Lipschitz continuous algorithms for covering problems, such as the minimum vertex cover, set cover, and feedback vertex set problems.
    
    Our algorithm for the feedback vertex set problem is based on linear programming, and in the rounding process, we develop and use a technique called cycle sparsification, which may be of independent interest.
\end{abstract}

\thispagestyle{empty}
\newpage

\thispagestyle{empty}
\tableofcontents

\setcounter{page}{0}
\newpage

\input{intro}
\input{vertex-cover}
\input{set-cover-naive}
\input{set-cover-greedy}
\input{set-cover-lp}

\input{set-cover-f}
\input{feedback-vertex-set}

\input{shared-randomness}

\bibliographystyle{abbrv}
\bibliography{main}

\appendix

\input{appendix}

\end{document}

%% file: intro.tex

\section{Introduction}

Combinatorial algorithms are widely used for decision-making and knowledge discovery, and it is important to ensure that their output remains stable even when subjected to small perturbations in the input.
As was discussed in~\cite{impagliazzo2022reproducibility,kumabe2023lipschitz,Varma2021}, failure to do so can lead to several problems, including costly decisions, reduced user trust, potential security concerns, and lack of replicability.
Unfortunately, many fundamental combinatorial algorithms are vulnerable to small perturbations to the input.
To address the impact of input perturbations on algorithms for weighted graph problems, Kumabe and Yoshida~\cite{kumabe2023lipschitz} recently introduced the concept of Lipschitz continuity of algorithms, and several fundamental graph problems, such as the minimum spanning tree problem, the shortest path problem, and the maximum matching problem, were studied from the viewpoint of Lipschitz continuity.
In this work, we further explore their approach and develop Lipschitz continuous algorithms for covering problems, or more specifically, the minimum vertex cover, set cover, and feedback vertex set problems.

\subsection{Lipschitz Continuity}

In this section, we formally define the Lipschitz continuity of an algorithm. 
We explain the definition using graph algorithms that generate a vertex set as its output, but the definition can be easily generalized to other combinatorial problems, such as the set cover problem (See Section~\ref{sec:set-cover-naive}).  

Consider a graph algorithm $\mathcal{A}$, that takes a graph $G=(V,E)$ and a weight vector $w \in \mathbb{R}_{\geq 0}^V$, and outputs a vertex set $S \subseteq V$.
To define the Lipschitz continuity of $\mathcal{A}$, we need to define distances between input weights and output sets.
For the input weights, we use the $\ell_1$ distance.
For the output sets, we follow the approach in~\cite{kumabe2023lipschitz}.
Namely, we map a vertex set $S \subseteq V$ to a vector $\sum_{v \in S}w_v \bm{1}_v$, where $\bm{1}_v \in \mathbb{R}^V$ is the characteristic vector of $v \in V$, that is, $\bm{1}_v(u) = 1$ if $u = v$ and $0$ otherwise.
Then, we use the $\ell_1$ distance between the mapped vectors to measure the distance between two outputs.
Specifically, for two vertex sets $S,S' \subseteq V$ and weight vectors $w,w'\in \mathbb{R}_{\geq 0}^V$, we use the distance defined as
\begin{align*}
    d_{\mathrm{w}}((S,w),(S',w')) := \left\|\sum_{v \in S}w_v \bm{1}_v - \sum_{v \in S'}w'_v \bm{1}_v\right\|_1 
    = \sum_{v \in S \cap S'}|w_v - w'_v| + \sum_{v \in S \setminus S'}w_v + \sum_{v \in S' \setminus S}w'_v.          
\end{align*}

Unfortunately, it is known that there is no deterministic Lipschitz continuous algorithm with respect to $d_{\mathrm{w}}$ for most problems~\cite{kumabe2023lipschitz}.
Hence, our focus is on randomized algorithms, and we need to define the distance between distributions over sets.
For two probability distributions $\mathcal{S}$ and $\mathcal{S}'$ over subsets of $V$, the \emph{earth mover's distance} between $\mathcal{S}$ and $\mathcal{S}'$ is defined as
\[
    \EMW\left((\mathcal{S},w),(\mathcal{S}',w')\right):=\min_{\mathcal{D}}\E_{(S,S')\sim \mathcal{D}} \left[d_{\mathrm{w}}\left((S,w),(S',w')\right)\right],
\]
where the minimum is taken over \emph{couplings} of $\mathcal{S}$ and $\mathcal{S}'$, that is, distributions over pairs of sets such that its marginal distributions on the first and second coordinates are equal to $\mathcal{S}$ and $\mathcal{S}'$, respectively.
We note that $\EMW$ coincides with $d_{\mathrm{w}}$ if the distributions $\mathcal{S}$ and $\mathcal{S}'$ are supported by single vertex sets.

For a randomized algorithm $\mathcal{A}$, a graph $G=(V,E)$, and a weight vector $w \in \mathbb{R}_{\geq 0}^V$, let $\mathcal{A}(G,w)$ denote the (random) output of $\mathcal{A}$ on $G$ and $w$.
Abusing the notation, we often identify the output with its distribution.
Then, we define the Lipschitz constant of a randomized algorithm as follows:
\begin{Definition}[Lipschitz constant of a randomized algorithm~\cite{kumabe2023lipschitz}]\label{def:randomized}
    Let $\mathcal{A}$ be a randomized algorithm that, given a graph $G=(V,E)$ and a weight vector $w \in \mathbb{R}_{\geq 0}^V$, outputs a (random) vertex set $\mathcal{A}(G,w) \subseteq V$.
    Then, the \emph{Lipschitz constant} of the algorithm $\mathcal{A}$ on a graph $G=(V,E)$ is
    \begin{align*}
        \sup_{\substack{w,w' \in \mathbb{R}_{\geq 0}^V,\\w\neq w'}}\frac{\EMW\left((\mathcal{A}(G,w),w),(\mathcal{A}(G,w'),w')\right)}{\|w-w'\|_1}.
    \end{align*}
    We say that $\mathcal{A}$ is \emph{Lipschitz continuous} if its Lipschitz constant is bounded for any graph $G=(V,E)$.
\end{Definition}

It is not hard to show that, if there exists an $\alpha$-approximation algorithm for a minimization problem on graphs, then for any $\epsilon > 0$, we can transform it to a $(1+\epsilon)\alpha$-approximation algorithm with Lipschitz constant $\epsilon^{-1} |V|$ (see Appendix~\ref{sec:naive-lipschitz-constant-bound} for details).
Hence, the question is whether we can design algorithms with Lipschitz constant sublinear in $|V|$.

Although Definition~\ref{def:randomized} is natural, it may not be apparent whether bounding the Lipschitz constant of a randomized algorithm is practically relevant because, even if the output distributions for two weighted graphs $(G,w)$ and $(G,w')$ are close, the outputs obtained by running the algorithm on $(G,w)$ and $(G,w')$ independently may be completely different.
Therefore, we consider randomized algorithms such that the outputs for two similar weighted graphs are close in expectation over the internal randomness.
More specifically, for a randomized algorithm $\mathcal{A}$, let $A_\pi$ denote the deterministic algorithm obtained from $\mathcal{A}$ by fixing its internal randomness to $\pi$.
Then, we define the Lipschitz constant with shared randomness of $\mathcal{A}$ as follows:
\begin{Definition}[Lipschitz constant of a randomized algorithm with shared randomness~\cite{kumabe2023lipschitz}]\label{def:randomized-shared-randomness}
    Let $\mathcal{A}$ be a randomized algorithm that, given a graph $G=(V,E)$ and a weight vector $w \in \mathbb{R}_{\geq 0}^V$, outputs a (random) vertex set $\mathcal{A}(G,w) \subseteq V$.
    Then, the \emph{Lipschitz constant under shared randomness} of the algorithm $\mathcal{A}$ on a graph $G=(V,E)$ is
    \begin{align*}
        \sup_{\substack{w,w' \in \mathbb{R}_{\geq 0}^V,\\w\neq w'}}\E_\pi\left[\frac{d_{\mathrm{w}}\left((\mathcal{A}_\pi(G,w),w),(\mathcal{A}_\pi(G,w'),w')\right)}{\|w-w'\|_1}\right],
    \end{align*}
    where $\pi$ represents the internal randomness of $\mathcal{A}$.
    We say that $\mathcal{A}$ is \emph{Lipschitz continuous under shared randomness} if its Lipschitz constant under shared randomness is bounded for any graph $G=(V,E)$.
\end{Definition}
Note that if a randomized algorithm has a Lipschitz constant $L$ under shared randomness, then it has Lipschitz constant $L$, and hence the former is a stronger property.

\subsection{Our Results}

\begin{table}[t!]
    \centering
    \caption{Summary of our results. For the minimum set cover problem, $n$ and $m$ denote the numbers of elements and sets, respectively, and $s$ and $f$ denote the maximum size of a set and maximum frequency of an element, respectively. For the minimum feedback vertex set problem, $n$ denotes the number of vertices. $\epsilon > 0$ is an arbitrary constant and $H_s = \sum_{i=1}^s 1/i$.}\label{tab:results}
    \begin{tabular}{llll}
        \toprule
        \multirow{2}{*}{Problem} & \multirow{2}{*}{Approximation ratio} & Lipschitz constant & \multirow{2}{*}{Reference} \\
        & & (under shared randomness) & \\
        \midrule
        Vertex Cover & $2$ & $4$ & Sec.~\ref{sec:vertex-cover} \\
        \midrule
        \multirow{4}{*}{Set Cover} & $s + \epsilon$ & $O(\epsilon^{-1}s^2)$ & Sec.~\ref{sec:set-cover-naive} \\
         & $ H_s + \epsilon$ & $\exp O(\epsilon^{-2}(s+\log f)s^2 \log^3 s)$ & Sec.~\ref{sec:set-cover-greedy} \\
        & $O(\log n)$ & $O(\sqrt{\min\{n,m\} \log m}\log n)$ & Sec.~\ref{sec:set-cover} \\
        & $f + \epsilon$ & $O(\epsilon^{-1}f^3 \sqrt{\min\{n,m\} \log m}\log n)$ & Sec.~\ref{sec:set-cover-f} \\
        \midrule
        Feedback Vertex Set & $O(\log n)$ & $O(\sqrt{n} \log^{3/2} n)$ & Sec.~\ref{sec:fedback-vertex-set} \\
        \bottomrule
    \end{tabular}
\end{table}

In this work, we provide Lipschitz continuous algorithms under shared randomness for various covering problems. 
Our results are summarized in Table~\ref{tab:results}.

For a graph $G=(V,E)$, we say that a set $S \subseteq V$ is a \emph{vertex cover} if every edge $e \in E$ is incident to a vertex in $S$.
In the minimum vertex cover problem, given a graph $G=(V,E)$ and a weight vector $w \in \mathbb{R}_{\geq 0}^V$, the goal is to find a vertex cover $S \subseteq V$ that minimizes the total weight $\sum_{v \in S}w_v$.
For the minimum vertex cover problem, we show a polynomial-time $2$-approximation algorithm with Lipschitz constant of at most $4$.
Given that computing $(2-\epsilon)$-approximate solution for the minimum vertex cover problem is NP-hard under the unique games conjecture for any $\epsilon > 0$~\cite{khot2008vertex}, our results show that we can reduce the Lipschitz constant to a small constant without sacrificing the approximation ratio.
We also prove that any exact (not necessarily polynomial-time) algorithm for the vertex cover problem has a Lipschitz constant of $\Omega(n)$, where $n$ is the number of vertices.

Let $U$ be a set of $n$ elements and $\mathcal{F}$ be a family of $m$ subsets of $U$.
Then, we say that $\mathcal{S} \subseteq \mathcal{F}$ is a \emph{set cover} if every element $e \in U$ is included in a set in $\mathcal{S}$.
In the minimum set cover problem, given a pair $(U,\mathcal{F})$ and a weight vector $w \in \mathbb{R}_{\geq 0}^\mathcal{F}$, the goal is to find a set cover $\mathcal{S} \subseteq \mathcal{F}$ that minimizes the total weight $\sum_{S \in \mathcal{S}}w_S$.
We provide three algorithms for the minimum set cover problem that are Lipschitz continuous with respect to the weight vector $w$.

First, we show that for any $\epsilon > 0$, there is a polynomial-time $(s+\epsilon)$-approximation algorithm with Lipschitz constant $O(\epsilon^{-1}s^2)$, where $s$ is the maximum set size in the instance, that is, $s := \max_{S \in \mathcal{F}}|S|$.
Our algorithm is obtained by modifying a naive $s$-approximation algorithm that selects, for each element $e \in U$, a set $S \subseteq \mathcal{F}$ including $e$ with the minimum weight, and adds it to the output.

Next, we show that for any $\epsilon > 0$, there is an  $(H_s+\epsilon)$-approximation algorithm with Lipschitz constant $\exp O(\epsilon^{-2}(s+\log f)s^2 \log^3 s)$, where $H_s = \sum_{i=1}^s (1/i) = O(\log s)$ and $f$ is the maximum frequency of an element, that is, $f := \max_{e \in U}|\{S \in \mathcal{F} : e \in S\}|$.
The time complexity is polynomial when $s = O(\log n)$.
It is known that the greedy algorithm achieves $H_s$-approximation~\cite{chvatal1979greedy,johnson1973approximation,lovasz1975ratio},
and we modify the algorithm to have a small Lipschitz constant when $s$ and $f$ are bounded.
Compared to the first algorithm, this algorithm has a better approximation ratio and a worse Lipschitz constant, which shows the trade-off between them.

Finally, we show that we can compute an $O(\log n)$- and $(f+\epsilon)$-approximate set cover in polynomial time with Lipschitz constant $O(\sqrt{\min\{n,m\} \log m}\log n)$ and $O(\epsilon^{-1}f^3 \sqrt{\min\{n,m\} \log m}\log n)$, respectively.
We note that it is NP-hard to achieve approximation ratio of $(1-\epsilon)\log n$ for any $\epsilon > 0$~\cite{dinur2014analytical} and NP-hard to achieve approximation ratio of $f-\epsilon$ for any $\epsilon > 0$ assuming the unique games conjecture~\cite{khot2008vertex}.
Hence, the approximation ratios of our algorithms are almost tight.
Also, the algorithms work even when $s$ and $f$ are not bounded.
We note that it is not hard to show that there is an $O((1+\epsilon)\log n)$-approximation (or $O((1+\epsilon)f)$-approximation) algorithm with Lipschitz constant $O(\epsilon^{-1}\min\{n,m\})$ (using the argument in Appendix~\ref{sec:naive-lipschitz-constant-bound}), and hence our algorithm (almost) quadratically improves the Lipschitz constant.

Given a graph $G=(V,E)$, a vertex set $S \subseteq V$ is called a \emph{feedback vertex set} if every cycle in $G$ has a vertex in $S$.
In the minimum feedback vertex set problem, given a graph $G=(V,E)$ and a weight vector $w \in \mathbb{R}_{\geq 0}^V$, the goal is to find a feedback vertex set $S \subseteq V$ that minimizes the total weight $\sum_{v \in S}w_v$.
We provide a polynomial-time algorithm for the minimum feedback problem whose performance is similar to the LP-based algorithm for the minimum set cover problem.
Specifically, we show that there exists a polynomial-time algorithm that outputs an $O(\log n)$-approximate feedback vertex set with Lipschitz constant $O(\sqrt{n}\log^{3/2} n)$, where $n$ is the number of vertices.
We mention that the minimum feedback vertex set problem admits $2$-approximation~\cite{bafna19992} and determining whether we can improve the multiplicative error of our algorithm to two (or a constant) is an interesting open problem.

\subsection{Proof Overview}

If an algorithm comprises multiple steps and the later steps depend on the results of earlier steps, the Lipschitz constant generally increases at each step, eventually reaching a large value. To achieve a small Lipschitz constant, the Lipschitz constant should not increase at each step or we should break the dependencies between steps. 
We provide an overview of our algorithm and how we achieve a small Lipschitz constant (without shared randomness) for each problem.
Then we discuss modifications required to achieve a small Lipschitz constant under shared randomness.

\paragraph{Vertex Cover}
Our algorithm is a primal-dual algorithm based on the following LP:

\begin{align*}
    \begin{array}{llll}
    \text{Primal}(G,w):= & \text{minimize} & \displaystyle \sum_{v \in V}w_v x_{v}, \\
    & \text{subject to} & x_{u}+x_{v} \geq 1 & \forall e=(u,v)\in E, \\
    & & x_{v} \geq 0 & \forall v \in V.
    \end{array}
\end{align*}
\begin{align*}
    \begin{array}{llll}
        \text{Dual}(G,w):= & \text{maximize} & \displaystyle \sum_{e \in E}y_{e}, \\
        & \text{subject to} & \displaystyle \sum_{e\in \delta(v)}y_e \leq w_v & \forall v\in V, \\
        & & y_e \geq 0 & \forall e \in E.
    \end{array}
\end{align*}
Here $\delta(v)$ denotes the set of edges incident to the vertex $v$.
The algorithm uses the concept of time to obtain feasible solutions for the dual problem. A vertex $v \in V$ is called \emph{tight} when it satisfies $\sum_{e\in \delta(v)}y_e=w_v$. At time $0$, we set $y_e=0$ for all edges $e \in E$, and simultaneously increment the value of $y_e$ by $1$ per unit of time as long as neither of the endpoints of $e$ is tight. 
Once all edges are incident to tight vertices, the algorithm adds each vertex $v$ to the output with probability $\sum_{e\in \delta(v)}y_e/w_v$. The approximation guarantee of $2$ follows from weak LP duality.

The main technical part lies in the analysis of the Lipschitz constant. Let $w$ and $w'$ be two different weight vectors. We denote the value of $y_e$ at time $t$ when the algorithm is executed with weights $w$ and $w'$ as $y_e(t)$ and $y'_e(t)$, respectively. We define the \emph{residual distance} as
\[
\mathrm{RD}(t) := \sum_{v\in V}\left|\left(w_v-\sum_{e\in \delta(v)}y_e(t)\right)-\left(w'_v-\sum_{e\in \delta(v)}y'_e(t)\right)\right|.
\]
We can prove that the residual distance does not increase over time.
It implies that for any given time $t \in \mathbb{R}_{\geq 0}$, $\sum_{e\in E}|y_e(t)-y'_e(t)|\leq |w-w'|$, which can be used to bound the Lipschitz constant.

\paragraph{Naive algorithm for Set Cover}

Since the minimum vertex cover problem is a special case of the minimum set cover problem, it is natural to expect that our Lipschitz continuous algorithm for the former can be extended to the latter. 
However, the proof of the monotone non-increasingness of the residual distance heavily relies on the fact that each element is included in exactly two sets, so this approach does not succeed. 
Therefore, we need to design new algorithms specific to the minimum set cover problem. 
We design multiple algorithms, and here we overview the simplest one.

We start with the observation that the algorithm that selects, for each element $e\in U$, a set $S \ni e$ with the minimum weight and outputs them is a trivial $s$-approximation algorithm, where $s$ is the maximum set size in $\mathcal{F}$. 
We make slight modifications to obtain a Lipschitz continuous algorithm. Specifically, instead of selecting the set with the minimum weight that contains each $e \in U$, we sample a set uniformly at random from the family of sets with their weight less than a constant times the minimum weight independently from others. 
By setting the constant to achieve an approximation ratio of $s + \epsilon$, we obtain the Lipschitz constant of $O(\epsilon^{-1}s^2)$ using the independence of the sampling process.

\paragraph{Greedy-based algorithm for Set Cover}

A well-known greedy algorithm for the minimum set cover problem is as follows: Start with an empty solution. At each step, let $R \subseteq U$ be the set of vertices still not covered, select the set $S \in \mathcal{F}$ that minimizes $w_S/|R \cap S|$, and add it to the solution. This algorithm is an $H_s$-approximation algorithm, where $H_s := \sum_{i=1}^s (1/i)$.

We modify this algorithm to obtain a Lipschitz continuous algorithm. 
First, we rewrite the operations of the original algorithm as follows: Define a multifamily $\mathcal{F}^\downarrow$ as $\bigcup_{S \in \mathcal{F}} (2^S \setminus \{\emptyset\})$, and define the weight $w^\downarrow_A$ for each element $A \in \mathcal{F}^\downarrow$ as the weight $w_S$ of the corresponding $S \in \mathcal{F}$.
By this rephrasing, each step of the greedy algorithm can be seen as the operation of selecting the set $A\subseteq \mathcal{F}^{\downarrow}$ of the minimum relative weight $w^{\downarrow}_A/|A|$ that has no intersection with already covered elements.
Then, we introduce a positive integer constant $M$ and round $w^\downarrow_A$ to an integer power of $s^{1/M}$.

The main trick of our algorithm is to apply the greedy algorithm after discarding some sets of $\mathcal{F}^\downarrow$. 
Specifically, we introduce a positive integer constant $K$ and discard a set $A \in \mathcal{F}^\downarrow$ if $\log_s w^\downarrow_A$ belongs to a specific interval of width $1$ modulo $K$. 
Then, whether to include $A$ in the solution only affects choices of having elements with weights less than $s^{K}\cdot w^\downarrow_A$, which breaks the dependencies between steps in the greedy algorithm.
Because we round the weights, there are at most $sKM$ possible values for $w^{\downarrow}_A/|A|$.
If we randomly decide whether to include elements with equal (rounded) weights in the solution, we can prove that the Lipschitz constant can be bounded by a constant determined by $s$ and $f$ raised to the power of $sKM$.

\paragraph{LP-based algorithm for Set Cover}

The general approach of solving an LP relaxation and then rounding the obtained fractional solution is powerful in designing approximation algorithms.
Unfortunately, it was not known whether this approach can be made Lipschitz continuous (with respect to the weighted distance).
In this work, we show a Lipschitz continuous algorithm that solves a covering LP and then use it to design Lipschitz continuous algorithms for various settings.


Let us begin by recalling the standard LP-based algorithm for the minimum set cover problem. 
Given an instance $(U,\mathcal{F})$ of the minimum set cover problem and a weight vector $w \in \mathbb{R}_{\geq 0}^{\mathcal{F}}$, we consider the following LP relaxation:
\begin{align}
    \begin{array}{lll}
        \text{minimize} & \displaystyle \sum_{S \in \mathcal{F}}w_S x_S, \\
        \text{subject to} & \displaystyle \sum_{S \in \mathcal{F}: e \in S}x_S \geq 1 & \forall e \in U, \\
        & 0 \leq x_S \leq 1 & \forall S \in \mathcal{F},
    \end{array}
    \label{eq:lp-set-cover-intro}
\end{align}
where $n = |U|$ is the number of elements.
Let $\OPTLP$ denote the optimal value of LP~\eqref{eq:lp-set-cover-intro} and $x \in \mathbb{R}^{\mathcal{F}}$ be the corresponding solution.
Then for each $S \in \mathcal{F}$, we independently add it to the solution with probability $x_S$, repeating this process $O(\log n)$ times to ensure that the output solution becomes a set cover with a high probability.
Clearly, the expected total weight of the obtained set cover is $O(\OPTLP \cdot \log n)$.

Now, we explain the modification required to make the algorithm above Lipschitz continuous. 
Suppose we have two weight vectors $w,w' \in \mathbb{R}_{\geq 0}^{\mathcal{F}}$ and their corresponding solutions to LP~\eqref{eq:lp-set-cover-intro} are denoted by $x$ and $x'$, respectively.
To achieve Lipschitz continuity in the rounding part, we need to ensure that $\sum_{S \in \mathcal{F}}|w_S x_S - w'_S x'_S|$ is small. 
To accomplish this, we introduce two regularizers in the objective function:
\[
    \sum_{S \in \mathcal{F}}w_S x_S + \frac{\lambda}{2} \sum_{S \in \mathcal{F}}x_S^2 + \frac{\kappa}{2} \sum_{S \in \mathcal{F}} \left(\frac{w_S x_S}{2W}+\frac{1}{2m}\right) \log \left(\frac{w_S x_S}{2W} + \frac{1}{2m}\right),
\]
where $m = |\mathcal{F}|$ is the number of sets, $\lambda = \Theta(\OPTLP / m)$, $\kappa = \Theta(\OPTLP/\log m)$, and $W = \Theta(\OPTLP)$\footnote{We denote the optimal LP value as $\OPTLP$ without specifying the weights, but this is an inaccurate representation since the optimal LP values for $w$ and $w'$ are different. However, we consider the situation where $w$ and $w'$ are close, and thus two optimal values are also close. Therefore, in this overview, we simply write $O(\OPTLP)$ to represent the amount proportional to both optimal values.}.
The addition of the first type of regularizer has been utilized in~\cite{Varma2021} to design an algorithm for the $s$-$t$ minimum cut problem with low average sensitivity, and in~\cite{kumabe2023lipschitz} to design a Lipschitz continuous algorithm for the bipartite maximum matching problem, where the distance between two matchings is computed without using the respective weights (as opposed to $\EMW$). 

To control the stability with respect to the weighted distance $d_{\mathrm{w}}(\cdot,\cdot)$, we also incorporate the second type of regularizer, which is the negative entropy of the vector $(\frac{w_S x_S}{2W}+\frac{1}{2m})_{S \in \mathcal{F}}$.
Exploiting strong convexity of these regularizers, respectively, we can show that $\sum_{S \in \mathcal{F}}|w_S x_S - w'_S x'_S| = O(\sqrt{m \log m} \cdot \|w-w'\|_1)$ while the additional error in the objective value is $O(\OPTLP)$.

An intuitive reason that we have a factor of $O(\sqrt{m \log m})$ in the above bound is as follows:
Assuming that the weight vector $w'$ is obtained from $w$ by increasing $w_S$ by $\delta>0$ and $\OPTLP = \Theta(1)$, the first regularizer guarantees that we have $\lambda (x_S - x'_S)^2 \leq \lambda \|x-x'\|_2^2 =O(\delta(x_S - x'_S))$, and hence we have $x_S - x'_S = O(\delta/\lambda)$.
Then, the second regularizer guarantees that $\kappa (\sum_{S \in \mathcal{F}}|w_S x_S - w'_S x'_S|)^2 = O(\delta (x_S-x'_S))$, and hence we have $\sum_{S \in \mathcal{F}}|w_S x_S - w'_S x'_S| = O(\delta/\sqrt{\kappa \lambda}) = O(\delta \sqrt{m \log m})$.

We can show that the Lipschitz constant for a single iteration of the rounding part is $\sum_{S \in \mathcal{F}}|w_S x_S - w'_S x'_S|/\|w-w'\|_1 = O(\sqrt{m \log m})$ using the independence of the sampling process.
Since we repeat the iteration $O(\log n)$ times, the final Lipschitz constant becomes $O(\sqrt{m \log m}\log n)$, while the expected total weight of the output is $O(\log n \cdot \OPTLP) = O(\log n \cdot \OPT)$.
Also, we note that we can improve the Lipschitz constant to $O(\sqrt{\min\{n,m\} \log m}\log n)$ by exploiting the fact that there exists an optimal solution with at most $\min\{n,m\}$ sets in $\mathcal{F}$.

By replacing the rounding scheme, we can also obtain a $(f+\epsilon)$-approximation algorithm with Lipschitz constant $O(\epsilon^{-1}f^3 \sqrt{m \log m}\log n)$, where $f$ is the maximum frequency of an element, which shows the versatility of the linear programming approach.

\paragraph{Feedback Vertex Set}

The minimum feedback vertex set problem can be represented as a minimum set cover problem, where vertices correspond to sets and cycles correspond to elements.
Therefore, given a graph $G=(V,E)$ and a weight vector $w \in \mathbb{R}_{\geq 0}^V$, we consider the LP relaxation~\eqref{eq:lp-set-cover-intro}, and we solve it using our algorithm for the minimum set cover problem with $\lambda = \Theta(\OPTLP/n)$ and $\kappa = \Theta(\OPTLP/\log n)$, where $n := |V|$ is the number of vertices (corresponding to $m$ in the set cover case).
This will yield an LP solution $x \in \mathbb{R}_{\geq 0}^V$ such that $\sum_{v \in V}w_v x_v = O(\OPTLP)$ and $\sum_{v \in C}x_v \geq 1$ for every cycle $C$ in $G$. (For technical reasons, we consider a tighter LP relaxation in the formal analysis.)
In addition, for LP solutions $x,x' \in \mathbb{R}^V$ constructed from weight vectors $w,w' \in \mathbb{R}^V$, respectively, we have $\sum_{v \in V}|w_v x_v - w'_v x'_v| = O(\sqrt{n \log n} \cdot \|w-w'\|_1)$.

However, an issue arises in the rounding part. 
To ensure that the output is a feedback vertex set, the independent sampling process needs to be repeated $O(\log |\mathcal{C}|)$ times, where $\mathcal{C}$ is the set of simple cycles in $G$.
However, because $|\mathcal{C}|$ can be exponential in the number of vertices in $G$, this rounding approach leads to polynomial approximation ratio and Lipschitz constant, which are excessively large.

To address this issue, we introduce a technique called \emph{cycle sparsification} and use it to round the LP solution.
For two weight vectors $z,\tilde{z} \in \mathbb{R}^V$, we say that $(G,\tilde{z})$ is an \emph{$\epsilon$-cycle sparsifier} of $(G,z)$ if
\[
    (1-\epsilon) z(C) \leq \tilde{z}(C) \leq (1 + \epsilon)z(C)
\]
holds for every cycle $C$ in $G$, where $z(C) = \sum_{v \in C} z_v$ and $\tilde{z}(C) = \sum_{v \in C} \tilde{z}_v$.
A similar notion for cuts has been extensively studied under the name of \emph{cut sparsifier}~\cite{benczur2015randomized,fung2019general}.

To obtain a cycle sparsifier, we modify the cut sparsification algorithm for integer-weighted graphs proposed by Fung~et~al.~\cite{fung2019general}.
In their algorithm, they compute the \emph{edge connectivity} $k_e$ of each edge $e \in E$, which is the minimum weight of a cut that separates the endpoints of $e$.
Then for each edge $e \in E$ with weight $z_e\in \mathbb{Z}_{\geq 0}$, they sample a number $c_e$ from the binomial distribution $\mathcal{B}(z_e, p_e)$, where $p_e = O(\log n / k_e)$, and assign the weight $c_e/p_e$ to it.
To construct a cycle sparsifier, we first compute the \emph{girth} $g_v$ of a vertex $v \in V$, the minimum weight of a cycle passing through the vertex.
Then for each vertex $v \in V$ with weight $z_e \in \mathbb{Z}_{\geq 0}$, we sample a number $c_v$ from the binomial distribution $\mathcal{B}(z_v, p_v)$, where  $p_v = O(\log n/g_v)$, and assign the weight $c_v/p_v$ to it.

In the analysis of the cut sparsifier proposed by Fung~et~al.~\cite{fung2019general}, they used a cut counting lemma, which asserts that, for any edge set $F \subseteq E$, the number of subsets of $F$ induced by small cuts is small.
We need a counterpart of this lemma for cycles, that is, we want to show that for any vertex set $S \subseteq V$, the number of subsets of $S$ induced by small cycles is small.
We prove this by generalizing the known argument for $S = V$~\cite{subramanian1995polynomial}.

Now, we explain how we use the cycle sparsification result for rounding.
Consider an LP solution $x \in \mathbb{R}_{\geq 0}^V$ for the minimum feedback vertex set problem.
Note that in the weighted graph $(G,x)$, every vertex has a girth of at least one. 
To round the LP solution $x$, we first construct an integer vector $z \in \mathbb{Z}_{\geq 0}^V$ by scaling it by a factor $n^2$.
Then, we apply cycle sparsification to $z$ using the girth lower bound of $n^2$.
The resulting vector $\tilde{z} \in \mathbb{R}_{\geq 0}^V$ is a cycle sparsifier with a high probability.
This implies that the set $S = \{v \in V: \tilde{z}_v > 0\}$ is a feedback vertex set because every cycle has a positive weight in $\tilde{z}$.
Furthermore, we can show that the total weight of $S$ is $O(\sum_{v \in V}w_v x_v \log n)$, which is $O(\log n \cdot \OPTLP)$.

Now let us consider Lipschitz continuity of the algorithm.
Let $x,x' \in \mathbb{R}^V$ be the LP solutions for weight vectors $w,w' \in \mathbb{R}^V$, respectively.
Then, considering that our cycle sparsification algorithm samples weights of vertices independently, we can show that it amplifies the Lipschitz constant merely by $O(\log n)$.
Combined with the bound $\sum_{v \in V}|w_v x_v - w'_v x'_v| = O(\sqrt{n \log n} \cdot \|w-w'\|_1)$, we can show that the Lipschitz constant of the entire algorithm is $O(\sqrt{n}\log^{3/2} n)$.

\paragraph{Lipschitz Continuity under Shared Randomness}
Our algorithms can straightforwardly be made Lipschitz continuous even in the shared randomness setting. 
Specifically, we replace the random sampling processes in our algorithms with stable sampling procedures developed by Kumabe and Yoshida~\cite{kumabe2023lipschitz} to ensure Lipschitz continuity under shared randomness.

\subsection{Related Work}

As mentioned earlier, Kumabe and Yoshida introduced the concept of Lipschitzness for graph algorithms~\cite{kumabe2023lipschitz}.
They devised Lipschitz continuous algorithms for the minimum spanning tree, shortest path, and maximum matching problems. 
This work expands the scope of Lipschitz continuous algorithms to covering problems.
Although different possible notions of Lipschitz continuity were considered, the ones used in this work (Definitions~\ref{def:randomized} and~\ref{def:randomized-shared-randomness}) are arguably the most natural.

The average sensitivity of algorithms, introduced by Varma and Yoshida~\cite{Varma2021}, is a concept closely related to Lipschitz continuity.
The \emph{average sensitivity} of an algorithm on an (unweighted) instance $X=(x_1,\ldots,x_n)$ consisting of $n$ elements is the average earth mover's distance between the output for $X$  and that for $X \setminus \{x_i\}$, where we take the expectation over $i \in \{1,2,\ldots,n\}$.
They showed algorithms for the minimum spanning tree, minimum cut, and maximum matching problems with low average sensitivity.
Since then algorithms with low average sensitivity have been proposed for various problems, including the minimum cut and  maximum matching problems~\cite{Yoshida2021}, $k$-means~\cite{yoshida2022average}, decision tree learning~\cite{hara2023average}, spectral clustering~\cite{peng2020average}.
We note that bounding the Lipschitz constant is usually more challenging than bounding the average sensitivity because it considers the worst-case input modification and requires dealing with arbitrarily small yet non-zero perturbations.

\subsection{Preliminaries}
We use bold symbols to denote random variables.
For a (continuous or discrete) set $X$, let $\mathcal{U}(X)$ denote the uniform distribution over $X$.
For a vector $w \in \mathbb{R}^V$ and a set $S \subseteq V$, we define $w(S) := \sum_{v \in S}w_v$.

The following is useful to bound the Lipschitz constant. 
The proof is quite similar to that of Lemma~1.7 of~\cite{kumabe2023lipschitz}.
\begin{lemma}\label{lem:seeoneelement}
Let $\mathcal{A}$ be an algorithm that takes a graph $G=(V,E)$ and a weight vector $w \in \mathbb{R}_{\geq 0}^V$ and outputs a subset of $V$.
Suppose that there exist some $c>0$ and $L>0$ such that
\[
    \EMW\left((\mathcal{A}(G,w),w), (\mathcal{A}(G,w+\delta \mathbf{1}_v),w+\delta \mathbf{1}_v)\right)\leq \delta L
\]
holds for any $v\in V$, $w\in \mathbb{R}_{\geq 0}^V$, and $\delta > 0$ with either $\delta\leq c\cdot w_v$ or $w_v=0$.
Then, $\mathcal{A}$ has Lipschitz constant $L$ on $G$.

Similarly, if the above inequality holds for any $v\in V$, $w\in \mathbb{R}_{\geq 0}^V$, and $0 < \delta < c$, then $\mathcal{A}$ has Lipschitz constant $L$ on $G$.
\end{lemma}

\subsection{Organization}

We provide Lipschitz continuous algorithm for the minimum vertex cover problem in Section~\ref{sec:vertex-cover}.
Then, we provide naive and greedy-based Lipschitz continuous algorithms for the minimum set cover problem in Sections~\ref{sec:set-cover-naive} and ~\ref{sec:set-cover-greedy}, respectively.
Then, we provide LP-based Lipschitz continuous algorithms for the minimum set cover problem with approximation ratio $O(\log n)$-approximation and $(f+\epsilon)$-approximation in Sections~\ref{sec:set-cover} and~\ref{sec:set-cover-f}, respectively.
We provide Lipschitz continuous algorithm for the feedback vertex set problem and the cycle sparsification result in Section~\ref{sec:fedback-vertex-set}.
Finally, we discuss modifications needed to bound the Lipschitz constant under shared randomness in Section~\ref{sec:shared-randomness}.

%% file: vertex-cover.tex
\section{Vertex Cover}\label{sec:vertex-cover}

The goal of this section is to show the following:
\begin{theorem}\label{thm:vertex-cover}
    There exists a polynomial-time $2$-approximation algorithm for the minimum vertex cover problem with Lipschitz constant $4$.
\end{theorem}
We first look at the primal and dual LP formulations for the minimum vertex cover problem. 
In the primal LP, the variable $x_v$ for a vertex $v \in V$ represents whether it belongs to a vertex cover.
\begin{align*}
    \begin{array}{llll}
        \text{Primal}(G,w):= & \text{minimize} & \displaystyle \sum_{v \in V}w_v x_{v}, \\
        & \text{subject to} & x_{u}+x_{v} \geq 1 & \forall e=(u,v)\in E,\\
        & & x_{v} \geq 0 & \forall v \in V.
    \end{array}
\end{align*}
\begin{align*}
    \begin{array}{llll}
        \text{Dual}(G,w):= & \text{maximize} & \displaystyle \sum_{e \in E}y_{e},\\
        & \text{subject to} & \sum_{e\in \delta(v)}y_e \leq w_v & \forall v\in V,\\
        && y_e \geq 0 & \forall e \in E.
    \end{array}
\end{align*}

\begin{algorithm}[t!]
\caption{Lipschitz continuous algorithm for the minimum vertex cover problem}\label{alg:vertexcover}
\Procedure{\emph{\Call{VertexCover}{$G,w$}}}{
    \KwIn{A graph $G=(V,E)$ and a weight vector $w \in \mathbb{R}_{\geq 0}^V$.}
    Let $y_e\leftarrow 0$ for all $e\in E$\;
    Let $A\leftarrow E$\;
    \While{$A\neq \emptyset$}{\label{line:vc-while}
        Let $\Delta t$ be the minimum value of $\frac{w_v-\sum_{e\in \delta(v)}y_e}{|\delta(v)\cap A|}$ among all non-tight vertices $v\in V$\;
        \For{$e\in A$}{
            $y_e\leftarrow y_e + \Delta t$\;
            Remove all edges incident to a tight vertex from $A$\;
        }
    }
    Let $\bm{C}\leftarrow \emptyset$\;
    \For{$v\in V$}{
        Sample $\bm{z}(v)$ uniformly from $[0,1]$\;
        \If{$w_v=0$ or $\bm{z}(v)\leq \frac{\sum_{e\in \delta(v)}y_e}{w_v}$}{
            Add $v$ to $\bm{C}$\;
        }
    }
    \Return $\bm{C}$\;
}
\end{algorithm}

Our algorithm, \Call{VertexCover}{}, is presented in Algorithm~\ref{alg:vertexcover}. 
In the while loop starting from Line~\ref{line:vc-while}, we call a vertex \emph{tight} if $w_v =\sum_{e\in \delta(v)}y_e$.
This algorithm is nearly identical to the standard primal-dual algorithm (see, e.g., Section 7.1 of~\cite{williamson2011design}) for the minimum vertex cover problem, and our main contribution lies in the analysis of the Lipschitz constant.

We analyze the approximation ratio and Lipschitz continuity of \Call{VertexCover}{} in Sections~\ref{subsec:vertex-cover-approximation-ratio} and~\ref{subsec:vertex-cover-lipschitz-continuity}, respectively.

\subsection{Approximation Ratio}\label{subsec:vertex-cover-approximation-ratio}

The following holds.
\begin{lemma}\label{lem:vertex-cover-approximation-ratio}
\Call{VertexCover}{} is a $2$-approximation algorithm.
\end{lemma}
\begin{proof}
Since removing vertices $v \in V$ with $w_v=0$ does not affect the algorithm and the optimal value, we can assume $w_v>0$ for every $v\in V$ to analyze approximation ratio without loss of generality.
At the end of each iteration in the while loop starting at Line~\ref{line:vc-while}, $A$ forms a set of edges such that at least one of the endpoints is tight. 
By the termination condition of the loop, at the end of the loop, at least one endpoint of each edge is tight. 
Since the output must include all tight vertices, Algorithm~\ref{alg:vertexcover} outputs a vertex cover.

Note that the value of $y_e$ corresponding to an edge $e \in E$ incident to a tight vertex does not increase further. 
Therefore, at the end of the algorithm, for every vertex $v \in V$, we have $\sum_{e\in \delta(v)}y_e \leq w_v$. 
In particular, $y$ is a feasible solution for $\mathrm{Dual}(G,w)$.
Let $\bm{C} = \Call{VertexCover}{G,w}$.
Then, we have
\begin{align*}
\E\left[\sum_{v \in \bm{C}} w_v \right]
&= \sum_{v\in V}\left(w_v\cdot \frac{\sum_{e\in \delta(v)}y_e}{w_v}\right)
= \sum_{v\in V}\left(\sum_{e\in \delta(v)}y_e\right)
=  2\sum_{e\in E}y_e
\leq 2\cdot \OPT,
\end{align*}
where the last inequality is because $y$ is a feasible solution for $\mathrm{Dual}(G,w)$ and the weak duality theorem of linear programming.
\end{proof}

\subsection{Lipschitz Continuity}\label{subsec:vertex-cover-lipschitz-continuity}

Now we evaluate the Lipschitz continuity. 
For two weight vectors $w,w'\in \mathbb{R}_{\geq 0}^{V}$, let us evaluate
\begin{align*}
\EMW\left((\Call{VertexCover}{G,w},w), (\Call{VertexCover}{G,w'},w')\right).
\end{align*}
Let $y_e$ and $y'_e$ denote the values of the variable $y_e$ in \Call{VertexCover}{$G,w$} and \Call{VertexCover}{$G,w'$}, respectively.
We introduce the concept of \emph{time} to the loop starting from Line~\ref{line:vc-while} and regard that the algorithm continuously increments the value of $y_e$ for each edge $e \in E$ by $1$ per unit time unless one of the endpoints is tight.
Let $y_e(t)$ and $y'_e(t)$ be the value of $y_e$ and $y'_e$ at time $t$ during the execution of \Call{VertexCover}{$G,w$} and \Call{VertexCover}{$G,w'$}, respectively. It can be easily observed that, for every $e \in E$, there exists a time $t_e$ such that $y_e(t)$ is expressed as follows:
\begin{align*}
y_e(t)=\left\{\begin{array}{cc}
t & (t\leq t_e)\\
t_e & (t_e < t).
\end{array}\right.
\end{align*}
We call the value 
\[
\mathrm{RD}(t) = \sum_{v\in V}\left|\left(w_v-\sum_{e\in \delta(v)}y_e(t)\right)-\left(w'(v)-\sum_{e\in \delta(v)}y'_e(t)\right)\right|
\]
the \emph{residual distance} between $y_e$ and $y'_e$ at time $t$. 
Let us evaluate how the residual distance changes during the loop starting at Line~\ref{line:vc-while}. 
At the beginning of the loop, the residual distance is $\mathrm{RD}(0)=\|w-w'\|_1$.

For a function $f:\mathbb{R} \to \mathbb{R}$, we denote the right derivative of $f$ by $\partial_{+}f(t)$, i.e., 
\begin{align*}
\partial_{+}f(t) := \lim_{x\rightarrow +0}\frac{f(t+x)-f(t)}{x}.
\end{align*}
The following lemma is the heart of our Lipschitz continuity analysis.
\begin{lemma}\label{lem:RDdecrease}
$\partial_{+}\mathrm{RD}(t)\leq 0$. \end{lemma}
\begin{proof}
Treating $y_e(t)$'s and $y'_e(t)$'s as functions of $t$, using an integer function $c_e: \mathbb{R}_{\geq 0} \to  \{-2,-1,0,1,2\}$ defined for each $e\in E$, the right derivative of $\mathrm{RD}(t)$ can be represented as
\begin{align*}
\partial_{+}\mathrm{RD}(t) = \sum_{e\in E}c_e(t)\left(\partial_{+}y_e(t) - \partial_{+}y'_e(t)\right),
\end{align*}
where $c_e(t)$ takes only values from $\{-2,-1,0,1,2\}$ because for each edge $e\in E$, the term $y_e(t)-y'_e(t)$ appears twice in the expression of $\mathrm{RD}(t)$.

Let $A(t)$ and $A'(t)$ denote the set $A$ at time $t$ in $\Call{VertexCover}{G,w}$ and $\Call{VertexCover}{G,w'}$, respectively.
Then, $\partial_{+}y_e(t)$ takes the value $1$ if $e\in A(t)$ and $0$ otherwise.
Similarly, $\partial_{+}y'_e(t)$ takes the value $1$ if $e\in A'(t)$ and $0$ otherwise. Using these facts, we have
\begin{align*}
\partial_{+}\mathrm{RD}(t) 
&= \sum_{e\in E}c_e\left(\partial_{+}y_e(t) - \partial_{+}y'_e(t)\right)\\
&= \sum_{e\in A(t)\cap A'(t)}c_e\cdot 0
+\sum_{e\in A(t)\setminus A'(t)}c_e\cdot 1
+\sum_{e\in A'(t)\setminus A(t)}c_e\cdot (-1)
+\sum_{e\in E\setminus (A(t)\cap A'(t))}c_e\cdot 0\\
&=\sum_{e\in A(t)\setminus A'(t)}c_e - \sum_{e\in A'(t)\setminus A(t)}c_e.
\end{align*}

Let $e=(u,v)\in A(t)\setminus A'(t)$. From the definitions of $A(t)$ and $A'(t)$, neither $u$ nor $v$ is tight at time $t$ in \Call{VertexCover}{$G,w$}, while in \Call{VertexCover}{$G,w'$}, at least one of them, say $u$, is tight. Therefore, we have
\begin{align*}
& \partial_{+}\left|\left(w_u-\sum_{e\in \delta(u)}y_e(t)\right)-\left(w'_u-\sum_{e\in \delta(u)}y'_e(t)\right)\right|
= \partial_{+}\left|\left(w_u-\sum_{e\in \delta(u)}y_e(t)\right)-0\right|\\
&= \partial_{+}\left(w_u-\sum_{e\in \delta(u)}y_e(t)\right)
= -\sum_{e\in \delta(u)}\partial_{+}y_e(t).
\end{align*}
In particular, at least once in the sum of $\partial_+\mathrm{RD}(t)$, the term $\partial_{+}y_e(t)$ makes a negative contribution. 
Additionally, because $y_e(t)-y'_e(t)$ appears twice in the expression of $\mathrm{RD}(t)$, we have $c_e\leq 0$.

By a similar argument, we obtain $c_e\geq 0$ for $e\in A'(t)\setminus A(t)$. Therefore, the claim holds.
\end{proof}

The following is immediate from Lemma~\ref{lem:RDdecrease}.
\begin{lemma}
At the end of the loop starting at Line~\ref{line:vc-while}, we have
\begin{align*}
\sum_{v\in V}\left|\sum_{e\in \delta(v)}y_e - \sum_{e\in \delta(v)}y'_e\right|\leq 2\|w-w'\|_1.
\end{align*}
\end{lemma}
\begin{proof}
Let $T$ be the later of the times $\Call{VertexCover}{G,w}$ and $\Call{VertexCover}{G,w'}$ terminate. 
Then, from Lemma~\ref{lem:RDdecrease}, we have $\mathrm{RD}(T)\leq \|w-w'\|_1$. Therefore,
\begin{align*}
\sum_{v\in V}\left|\sum_{e\in \delta(v)}y_e(T) - \sum_{e\in \delta(v)}y'_e(T)\right|
&\leq \sum_{v\in V}\left(\left|\left(w_v-\sum_{e\in \delta(v)}y_e(t)\right)-\left(w'_v-\sum_{e\in \delta(v)}y'_e(t)\right)\right| + \left|w_v-w'_v\right|\right)\\
&= \mathrm{RD}(T)+\|w-w'\|_1\leq 2\|w-w'\|_1.\qedhere
\end{align*}
\end{proof}

\begin{lemma}\label{lem:vertex-cover-lipschitz-continuity}
\Call{VertexCover}{} has Lipschitz constant $4$.
\end{lemma}
\begin{proof}
By Lemma~\ref{lem:seeoneelement}, we can assume there exists a vertex $u\in V$ such that $w_v=w'_v$ holds for $v\in V\setminus \{u\}$ and $w_u<w'_u$. 
To bound the earth mover’s distance between the output vertex covers for $(G, w)$ and $(G, w')$,
we consider the coupling $\mathcal{W}$ between $\bm{z}$ and $\bm{z}'$, defined as follows: For each $z \in [0,1]^V$, we transport the probability mass for $\bm{z} = z$ to that for $\bm{z}'=z$. 
When $w_u=0$, we have
\begin{align*}
&\EMW\left((\Call{VertexCover}{G,w},w), (\Call{VertexCover}{G,w'},w')\right)\\
&\leq \sum_{v\in V\setminus \{u\}\colon w_v > 0}\left|\frac{\sum_{e\in \delta(v)}y_e}{w_v}-\frac{\sum_{e\in \delta(v)}y'_e}{w_v}\right|\cdot w_v
+ \frac{\sum_{e\in \delta(u)}y'_e}{w'_u}\cdot w'_u \\
&\leq \sum_{v\in V\setminus \{u\}\colon w_v > 0}\left|\sum_{e\in \delta(v)}y_e - \sum_{e\in \delta(v)}y'_e\right|
+ 1\cdot w'_u\\
&\leq \sum_{v\in V}\left|\sum_{e\in \delta(v)}y_e - \sum_{e\in \delta(v)}y'_e\right|+\|w-w'\|_1
\leq 3\|w-w'\|_1,
\end{align*}
where the first inequality is because $w_v\neq w'_v$ holds for $u=v$ and $w_u=0$.
Otherwise, we have
\begin{align*}
&\EMW\left((\Call{VertexCover}{G,w},w), (\Call{VertexCover}{G,w'},w')\right)\\
&\leq \sum_{v\in V\setminus \{u\}\colon w_v > 0}\left|\frac{\sum_{e\in \delta(v)}y_e}{w_v}-\frac{\sum_{e\in \delta(v)}y'_e}{w_v}\right|\cdot w_v\\
&\quad + \min\left(\frac{\sum_{e\in \delta(u)}y_e}{w_u},\frac{\sum_{e\in \delta(u)}y'_e}{w'_u}\right)\cdot (w'_u-w_u) + \left|\frac{\sum_{e\in \delta(u)}y_e}{w_u}-\frac{\sum_{e\in \delta(u)}y'_e}{w'_u}\right|\cdot w'_u 
\\
&\leq \sum_{v\in V\setminus \{u\}\colon w_v > 0}\left|\sum_{e\in \delta(v)}y_e - \sum_{e\in \delta(v)}y'_e\right|
+ 1\cdot (w'_u-w_u) + \left|\frac{\sum_{e\in \delta(u)}y_e}{w_u}-\frac{\sum_{e\in \delta(u)}y'_e}{w'_u}\right|\cdot w'_u \\
&\leq \sum_{v\in V\setminus \{u\}\colon w_v > 0}\left|\sum_{e\in \delta(v)}y_e - \sum_{e\in \delta(v)}y'_e\right|
+(w'_u-w_u)\\
&\quad +\left(\left|\frac{\sum_{e\in \delta(u)}y_e}{w_u}-\frac{\sum_{e\in \delta(u)}y_e}{w'_u}\right|+\left|\frac{\sum_{e\in \delta(u)}y_e}{w'_u}-\frac{\sum_{e\in \delta(u)}y'_e}{w'_u}\right|\right)\cdot w'_u\\
&\leq \sum_{v\in V\setminus \{u\}\colon w_v > 0}\left|\sum_{e\in \delta(v)}y_e - \sum_{e\in \delta(v)}y'_e\right|
+(w'_u-w_u)+\left|\frac{w_u}{w_u}-\frac{w_u}{w'_u}\right|\cdot w'_u+\left|\sum_{e\in \delta(u)}y_e-\sum_{e\in \delta(u)}y'_e\right|\\
&\leq \sum_{v\in V\colon w_v > 0}\left|\sum_{e\in \delta(v)}y_e - \sum_{e\in \delta(v)}y'_e\right|
+2(w'_u-w_u)\\
&\leq 2\|w-w'\|_1 + 2\|w-w'\|_1 = 4\|w-w'\|_1,
\end{align*}
where the first inequality is because $w_v = w'_v$ holds for $v\neq u$.
The contribution of vertex $u$ to the earth mover's distance is
\begin{align*}
& \frac{\sum_{e\in \delta(u)}y_e}{w_u}\cdot |w'_u-w_u| + \left(\frac{\sum_{e\in \delta(u)}y'_e}{w'_u}-\frac{\sum_{e\in \delta(u)}y_e}{w_u}\right)\cdot w'_u\\
&\leq \min\left(\frac{\sum_{e\in \delta(u)}y_e}{w_u},\frac{\sum_{e\in \delta(u)}y'_e}{w'_u}\right)\cdot (w'_u-w_u) + \left|\frac{\sum_{e\in \delta(u)}y_e}{w_u}-\frac{\sum_{e\in \delta(u)}y'_e}{w'_u}\right|\cdot w'_u 
\end{align*}
for $\frac{\sum_{e\in \delta(u)}y_e}{w_u}\leq \frac{\sum_{e\in \delta(u)}y'_e}{w'_u}$ and
\begin{align*}
& \frac{\sum_{e\in \delta(u)}y'_e}{w'_u}\cdot |w'_u-w_u| + \left(\frac{\sum_{e\in \delta(u)}y_e}{w_u}-\frac{\sum_{e\in \delta(u)}y'_e}{w'_u}\right)\cdot w_u\\
&\leq \min\left(\frac{\sum_{e\in \delta(u)}y_e}{w_u},\frac{\sum_{e\in \delta(u)}y'_e}{w'_u}\right)\cdot (w'_u-w_u) + \left|\frac{\sum_{e\in \delta(u)}y_e}{w_u}-\frac{\sum_{e\in \delta(u)}y'_e}{w'_u}\right|\cdot w'_u 
\end{align*}
otherwise.

Therefore, the Lipschitz constant of Algorithm~\ref{alg:vertexcover} is
\begin{align*}
\sup_{\substack{w,w' \in \mathbb{R}_{\geq 0}^V,\\w\neq w'}}\frac{\EMW\left((\Call{VertexCover}{G,w},w), (\Call{VertexCover}{G,w'},w')\right)}{\|w-w'\|_1} 
\leq  \sup_{\substack{w,w' \in \mathbb{R}_{\geq 0}^V,\\w\neq w'}}\frac{4\|w-w'\|_1}{\|w-w'\|_1} 
= 4.
\end{align*}
Hence, the claim holds.
\end{proof}

Theorem~\ref{thm:vertex-cover} follows by combining Lemmas~\ref{lem:vertex-cover-approximation-ratio} and~\ref{lem:vertex-cover-lipschitz-continuity}.

\subsection{Lower Bound}

Here, we give a lower bound on the Lipschitz constant for the minimum vertex cover problem.
\begin{lemma}
For $\epsilon>0$, any (not necessarily polynomial-time) exact algorithm for the minimum vertex cover problem has Lipschitz constant $\Omega(|V|)$.
\end{lemma}
\begin{proof}
Let $n$ be a positive integer and $G = (V, E)$ be a complete bipartite graph, where each part $V_1, V_2$ has size $n$. 
Clearly, the minimum vertex cover for the weight vector $w \in \mathbb{R}_{\geq 0}^V$ is the part with the smaller total weight. Let $v_1 \in V_1$ and $v_2 \in V_2$.
Now, consider the weight vectors $w, w' \in \mathbb{R}_{\geq 0}^V$ given by
\begin{align*}
    w_v=\left\{\begin{array}{cc}
        1 & (v\neq v_1)\\
        0 & (v=v_1)
    \end{array}\right.,\quad 
    w'_v=\left\{\begin{array}{cc}
        1 & (v\neq v_2)\\
        0 & (v=v_2)
    \end{array}\right..
\end{align*}
Since any exact algorithm outputs $V_1$ for the weight $w$ and $V_2$ for the weight $w'$, the Lipschitz constant is at least
\[
\frac{\|\bm{1}_{V_1\setminus \{v_1\}}-\bm{1}_{V_2\setminus \{v_2\}}\|_1}{\|w-w'\|_1}=\frac{2n-2}{2}\geq \Omega(|V|).
\qedhere
\]
\end{proof}

%% file: set-cover-naive.tex

\section{Naive Algorithm for Set Cover}\label{sec:set-cover-naive}

We first formally define the Lipschitz constant of an algorithm for the set cover problem:
\begin{Definition}[Lipschitz constant of a randomized algorithm for the set cover problem]\label{def:randomized-set-cover}
    Let $\mathcal{A}$ be a randomized algorithm that, given an instance $(U,\mathcal{F})$ of the set cover problem and a weight vector $w \in \mathbb{R}_{\geq 0}^{\mathcal{F}}$, outputs a (random) set cover $\mathcal{A}(U,\mathcal{F},w) \subseteq \mathcal{F}$.
    Then, the \emph{Lipschitz constant} of the algorithm $\mathcal{A}$ on an instance $(U,\mathcal{F})$ is
    \begin{align*}
        \sup_{\substack{w,w' \in \mathbb{R}_{\geq 0}^{\mathcal{F}},\\w\neq w'}}\frac{\EMW\left((\mathcal{A}(U,\mathcal{F},w),w),(\mathcal{A}(U,\mathcal{F},w'),w')\right)}{\|w-w'\|_1}.
    \end{align*}
    We say that $\mathcal{A}$ is \emph{Lipschitz continuous} if its Lipschitz constant is bounded for any instance $(U,\mathcal{F})$.
\end{Definition}

The goal of this section is to show the following:
\begin{theorem}\label{thm:set-cover-naive}
    For any $\epsilon > 0$, there exists a polynomial-time $(s+\epsilon)$-approximation algorithm for the minimum set cover problem with Lipschitz constant $O(\epsilon^{-1}s^2)$, where $s$ is the maximum size of a set in the instance.
\end{theorem}

Algorithm, \Call{NaiveSetCover}{}, is presented in Algorithm~\ref{alg:simplesetcover_const}. 
The idea is simple: For each element $e\in U$, we sample one set containing $e$ with nearly minimum weight, and collect these sets to form the output.

\begin{algorithm}[t!]
\caption{Naive Lipschitz continuous algorithm for the minimum set cover problem}\label{alg:simplesetcover_const}
\Procedure{\emph{\Call{NaiveSetCover}{$U,\mathcal{F}, w, \epsilon$}}}{
    \KwIn{A set $U$, $\mathcal{F} \subseteq 2^U$, a weight vector $w \in \mathbb{R}_{\geq 0}^{\mathcal{F}}$, and $\epsilon > 0$. }
    $s \gets \max_{S \in \mathcal{F}}|S|$\;
    Sample $\bm{b}$ uniformly from $[0,1]$\;\label{line:simplesetcover_sampleb}
    Let $\bm{\mathcal{C}}\leftarrow \emptyset$\;
    \For{$e\in U$}{
        \If{$\min_{S\in \mathcal{F}\colon e\in S}w_S=0$}{
            Let $\bm{t}_e=-\infty$ and $\bm{X}_{e,\bm{t}_e}=\left\{S \in \mathcal{F} : e \in S \wedge w_S=0\right\}$
        }
        \Else{
            Let $\bm{t}_e$ be the minimum integer such that the set $\bm{X}_{e,\bm{t}_e}=\left\{S \in \mathcal{F} : e \in S \wedge \left(1+\frac{\epsilon}{s}\right)^{\bm{b}+\bm{t}_e}\leq w_S < \left(1+\frac{\epsilon}{s}\right)^{\bm{b}+\bm{t}_e+1}\right\}$ is nonempty\;
        }
        Sample $\bm{S}_e$ uniformly from $\bm{X}_{e,\bm{t}_e}$ and add it to $\bm{\mathcal{C}}$\;       
    }
    \Return $\bm{\mathcal{C}}$\;
}
\end{algorithm}

We begin by analyzing the approximation ratio.
\begin{lemma}\label{lem:set-cover-naive-approximation}
\Call{NaiveSetCover}{} has approximation ratio $s+\epsilon$.
\end{lemma}
\begin{proof}
Let $\mathcal{F}^*$ be a minimum set cover of $(U,\mathcal{F},w)$. Then, we have
\begin{align*}
& w(\bm{\mathcal{C}})
\leq \sum_{e\in U}\left(1+\frac{\epsilon}{s}\right)\min_{S\ni e}w_S
=(s+\epsilon)\sum_{e\in U}\min_{S\ni e}\frac{w_S}{s}\\
&\leq (s+\epsilon)\sum_{e\in U}\sum_{S\ni e\colon S\in \mathcal{F}^*}\frac{w_S}{s}
= (s+\epsilon)\sum_{S\in \mathcal{F}^*}\sum_{e\in S}\frac{w_S}{s}
\leq (s+\epsilon)\sum_{S\in \mathcal{F}^*}w_S = (s+\epsilon)\OPT,
\end{align*}
where the last inequality is because $|S|\leq s$ holds for all $S\in \mathcal{F}$.
\end{proof}

Next we analyze the Lipschitz continuity.
For a weight vector $w \in \mathbb{R}_{\geq 0}^{\mathcal{F}}$, let $w' = w + \delta \bm{1}_T$ for $T \in \mathcal{F}$.
Then by Lemma~\ref{lem:seeoneelement}, it suffices to analyze
\begin{align*}
\EMW\left((\Call{NaiveSetCover}{U,\mathcal{F},w,\epsilon},w), (\Call{NaiveSetCover}{U,\mathcal{F},w', \epsilon},w')\right).
\end{align*}
Moreover, we can assume $\delta \leq w_T$ or $w_T=0$.

Let $\bm{b}$ and $\bm{b}'$ be the parameter $\bm{b}$ sampled in $\Call{NaiveSetCover}{U,\mathcal{F},w,\epsilon}$ and $\Call{NaiveSetCover}{U,\mathcal{F},w', \epsilon}$, respectively. 
We use similar notations for $\bm{t}_e$, $\bm{X}_{e,\bm{t}_e}$ and $\bm{S}_e$.
To bound the earth mover’s distance, we consider the coupling $\mathcal{W}$ between $\left(\bm{b}, (\bm{S}_e)_{e\in U}\right)$ and $\left(\bm{b}', (\bm{S}'_e)_{e\in U}\right)$ defined as follows: for each $b\in [0,1]$, the probability mass for $\bm{b} = b$ is transported to that for $\bm{b}' = b$. For the probability mass with $\bm{b}=\bm{b}'=b$, for each $e \in E$ and $S\in \mathcal{F}$, the probability mass for $\bm{S}_e=S$ is transported to that for $\bm{S}'_e=S$ as much as possible. 
Suppose $\left(\left(\bm{b}, (\bm{S}_e)_{e\in U}\right),\left(\bm{b}', (\bm{S}'_e)_{e\in U}\right)\right)\sim \mathcal{W}$. Then, $\bm{S}_e=\bm{S}'_e$ holds for $e\in U\setminus T$. The following lemma bounds the probability that $\bm{S}_e=\bm{S}'_e$ holds for $e \in T$.

\begin{lemma}\label{lem:simplesetcover_prob}
Let $e\in T$. 
If $w_T>0$, then we have 
\begin{align*}
\Pr_{\mathcal{W}}\left[\bm{S}_e\neq \bm{S}'_e\right]\leq \frac{2s\delta }{\epsilon w_T}.
\end{align*}
\end{lemma}
\begin{proof}
We have
\begin{align*}
\Pr_{\mathcal{W}}\left[\bm{S}_e\neq \bm{S}'_e\right]
&\leq \Pr_{b\sim \mathcal{U}([0,1])}\left[\exists t, \bm{X}_{e,t_e}\neq \bm{X}'_{e,t_e}\mid \bm{b}=\bm{b}'=b\right]\\
&= \Pr_{b\sim \mathcal{U}([0,1])}\left[\floor{\log_{1+\epsilon/s}w_T-b}\neq \floor{\log_{1+\epsilon/s}(w_T+\delta)-b}\right]\\
&= \left(\log_{1+\epsilon/s}(w_T+\delta) - \log_{1+\epsilon/s} w_T\right)
= \log_{1+\epsilon/s}\left(1+\frac{\delta}{w_T}\right)\\
&\leq \frac{\delta}{w_T\log (1+\epsilon/s)}
\leq \frac{\delta}{w_T}\cdot \frac{2s}{\epsilon},
\end{align*}
where the third inequality is from $\log(1+x)\leq x$ and the last inequality is from $\log\left(1+\epsilon/s\right)\geq \epsilon/(2s)$ for $\epsilon/s\leq  1$.
\end{proof}

Using this, we can bound the Lipschitz constant of \Call{NaiveSetCover}{}.
\begin{lemma}\label{lem:set-cover-naive-lipschitz}
\Call{NaiveSetCover}{} has Lipschitz constant $O(\epsilon^{-1}s^2)$.
\end{lemma}
\begin{proof}
For each $e\in T$, we have 
\begin{align*}
w_{\bm{S}_e} &\leq \left(1+\frac{\epsilon}{s}\right)\min_{S\ni e}w_S\leq \left(1+\frac{\epsilon}{s}\right)w_T, \text{ and}\\
w'_{\bm{S}'_e} &\leq \left(1+\frac{\epsilon}{s}\right)\min_{S\ni e}w'_S\leq \left(1+\frac{\epsilon}{s}\right)w'_T=\left(1+\frac{\epsilon}{s}\right)(w_T+\delta).
\end{align*}
Assume $w_T=0$. Then, we have
\begin{align*}
&\EMW\left((\Call{NaiveSetCover}{U,\mathcal{F},w,\epsilon},w), (\Call{NaiveSetCover}{U,\mathcal{F},w+\delta \bm{1}_T, \epsilon},w+\delta \bm{1}_T)\right)\\
&\leq \E_{\mathcal{W}}\left[\sum_{e\in U, \bm{S}_e\neq \bm{S}'_e}\left(w_{\bm{S}_e}+w'_{\bm{S}'_e}\right)\right]
\leq \E_{\mathcal{W}}\left[\sum_{e\in T }\left(w_{\bm{S}_e}+w'_{\bm{S}'_e}\right)\right]
\leq s\cdot \left(0+\left(1+\frac{\epsilon}{s}\right)\delta\right) \leq 2s\delta \leq 12s^2\epsilon^{-1}\delta,
\end{align*}
where the third inequality is from $w_T=0$ and the last inequality is from $\frac{\epsilon}{s}\leq 1$.
Otherwise, we have
\begin{align*}
&\EMW\left((\Call{NaiveSetCover}{U,\mathcal{F},w,\epsilon},w), (\Call{NaiveSetCover}{U,\mathcal{F},w+\delta \bm{1}_T, \epsilon},w+\delta \bm{1}_T)\right)\\
&\leq \E_{\mathcal{W}}\left[\sum_{e\in U, \bm{S}_e\neq \bm{S}'_e}\left(w_{\bm{S}_e}+w'_{\bm{S}'_e}\right)\right]\\
&\leq \sum_{e\in T} \left(\Pr_{\mathcal{W}}\left[\bm{S}_{e}\neq \bm{S}'_e\right]\cdot \E_{\mathcal{W}}\left[w_{\bm{S}_e}+w'_{\bm{S}'_e}\mid \bm{S}_e\neq \bm{S}'_e\right]\right) \\ 
&\leq \sum_{e\in T} \left(\Pr_{\mathcal{W}}\left[\bm{S}_{e}\neq \bm{S}'_e\right]\cdot \left(1+\frac{\epsilon}{s}\right)(w_T+(w_T+\delta))\right)\\
&\leq s\cdot \frac{2s\delta}{\epsilon w_T}\cdot \left(1+\frac{\epsilon}{s}\right)\cdot 3w_T
\leq \frac{12s^2}{\epsilon}\cdot \delta ,
\end{align*}
where the fourth inequality is from Lemma~\ref{lem:simplesetcover_prob} and $\delta\leq w_T$, and the last inequality is from $\epsilon/s\leq 1$.
Therefore, Algorithm~\ref{alg:simplesetcover_const} has Lipschitz constant $\frac{12\epsilon^{-1}s^2 \cdot \delta}{\delta}=O(\epsilon^{-1}s^2)$.
\end{proof}

Theorem~\ref{thm:set-cover-naive} follows by combining Lemmas~\ref{lem:set-cover-naive-approximation} and~\ref{lem:set-cover-naive-lipschitz}.

%% file: set-cover-greedy.tex
\section{Greedy-based Algorithm for Set Cover}\label{sec:set-cover-greedy}

In this section, we provide an algorithm for the set cover problem where the size of each set is at most $s$ and each element appears in at most $f$ sets. 
Specifically, we show the following:
\begin{theorem}\label{thm:set-cover-greedy}
For any $\epsilon\in (0,1]$, there exists an $\left(H_s+\epsilon\right)$-approximation algorithm for the minimum set cover problem with Lipschitz constant $\exp{O\left(\epsilon^{-2}\left(s+\log f\right)s^2\log^3 s\right)}$, where $H_s := \sum_{i=1}^{s}(1/i)$. 
The time complexity is polynomial when $s = O(\log n)$, where $n$ is the number of elements.
\end{theorem}

\subsection{Algorithm Description}

\begin{algorithm}[t!]
\caption{Greedy algorithm for the minimum set cover problem~\cite{chvatal1979greedy,johnson1973approximation,lovasz1975ratio}}\label{alg:setcover_known}
\Procedure{\emph{\Call{ClassicalGreedy}{$U,\mathcal{F},w$}}}{
    \KwIn{A set $U$, $\mathcal{F} \subseteq 2^U$, and a weight vector $w \in \mathbb{R}_{\geq 0}^{\mathcal{F}}$.}
    $\mathcal{C}\leftarrow \emptyset, R\leftarrow V$\;
    \While{$R\neq \emptyset$}{\label{line:setcoverknown_loop}
        Let $S=\mathrm{argmin}_{S' \in \mathcal{F}\colon R\cap S'\neq \emptyset}w_{S'}/|R\cap S'|$\;\label{line:setcoverknown_chooseS}
        Add $S$ to $\mathcal{C}$\;
        $R\leftarrow R\setminus S$\;
    }
    \Return $\mathcal{C}$\;
}
\end{algorithm}

\begin{algorithm}[t!]
\caption{Lipschitz continuous algorithm for the minimum set cover problem}\label{alg:setcover_const}
\Procedure{\emph{\Call{LipschitzGreedy}{$U,\mathcal{F}, w,K, M$}}}{
    \KwIn{A set $U$, $\mathcal{F} \subseteq 2^U$, a weight vector $w \in \mathbb{R}_{\geq 0}^{\mathcal{F}}$, and $K,M \in \mathbb{Z}_{\geq 1}$.}
    $s \gets \max_{S \in \mathcal{F}}|S|$\;
    Sample $\bm{b}$ uniformly from $[0,1]$\;\label{line:setcover_sampleb}
    Let $\bm{\pi}$ be a random permutation over $\{A \subseteq U\colon |A|\leq s\}$\;\label{line:setcoverconst_samplepi}
    $\bm{Q}\leftarrow \emptyset$\;
    \For{$S\in \mathcal{F}$}{
        Let $\bm{i}_S$ be the unique integer with $s^{K\bm{b}+\bm{i}_S/M}\leq w_S < s^{K\bm{b}+(\bm{i}_S+1)/M}$ if $w_S>0$, and $\bm{i}_S=-\infty$ otherwise\;\label{line:setcoverconst_is}
        \For{$\emptyset\neq A\subseteq S$}{
            \If{$\bm{i}_S>-\infty$ and $\bm{i}_S \bmod KM \geq M \log_s |A|$\label{line:setcoverconst_add_condition} }{
                Add $\left(\bm{i}_S- M \log_s |A|, A, S\right)$ to $\bm{Q}$\;\label{line:setcoverconst_add}
            }\ElseIf{$\bm{i}_S=-\infty$}{
                Add $\left(-\infty, A, S\right)$ to $\bm{Q}$\;            
            }
        }
    }
    $\bm{\mathcal{C}}\leftarrow \emptyset, \bm{R}\leftarrow U$\;
    \For{$(x,A,S)\in \bm{Q}$ in increasing order of $x$, where ties are broken according to $\pi(A)$}{\label{line:setcoverconst_loop}
        \If{$A\subseteq \bm{R}$}{
            Add $S$ to $\bm{\mathcal{C}}$\;
            Remove all elements in $A$ from $\bm{R}$\;
        }
    }
    \Return $\bm{\mathcal{C}}$\;
}
\end{algorithm}

Before explaining the algorithm, let us first introduce the $H_s$-approximation algorithm for the set cover problem given by Johnson~\cite{johnson1973approximation}, Lov\'{a}sz~\cite{lovasz1975ratio} and Chv\'{a}tal~\cite{chvatal1979greedy}.
The $H_s$-approximation algorithm in Algorithm~\ref{alg:setcover_known}  follows a greedy approach where sets are sequentially added to the solution starting from an empty set. At each step of the greedy algorithm, denoting the set of currently uncovered elements as $R \subseteq U$, the algorithm selects a set $S \in \mathcal{F}$ with $R\cap S\neq \emptyset$ that minimizes $w_S/|R\cap S|$ and adds it to the solution.

To evaluate the Lipschitz continuity, let us compare the processes of \Call{ClassicalGreedy}{$U,\mathcal{F},w$} and \Call{ClassicalGreedy}{$U,\mathcal{F},w+\delta \bm{1}_T$} for $\delta>0$ and $T\in \mathcal{F}$. 
Let $S_i$ and $S'_i$ be the set selected during the $i$-th iteration of the loop starting from Line~\ref{line:setcoverknown_loop} for \Call{ClassicalGreedy}{$U,\mathcal{F},w$} and \Call{ClassicalGreedy}{$U,\mathcal{F},w+\delta \bm{1}_T$}, respectively. 
Let $R_i$ and $R'_i$ be the $R$ sets at the end of that iteration for \Call{ClassicalGreedy}{$U,\mathcal{F},w$} and \Call{ClassicalGreedy}{$U,\mathcal{F},w+\delta \bm{1}_T$}, respectively. 
Assume there exists $i$ such that $T=S_i$. 
Up until the $i$-th iteration, \Call{ClassicalGreedy}{} behaves the same for both instances, and particularly, $R_j=R'_j$ for $j<i$. 
However, since $S_i$ may not be equal to $S'_i$, $R_i$ may not be equal to $R'_i$. 
Thus, the value of $|R\cap S|$ in the $(i+1)$-st iteration of \Call{ClassicalGreedy}{$U,\mathcal{F},w$} may differ from its value in \Call{ClassicalGreedy}{$U,\mathcal{F},w+\delta \bm{1}_{T}$} (even for $S\neq T$), and we cannot guarantee $S_{i+1} = S'_{i+1}$, causing the algorithm's behavior to diverge for the two instances.

To address this problem, we modify the criterion for selecting a set $S$ to be added to set $\mathcal{C}$ in the loop starting from Line~\ref{line:setcoverknown_chooseS} in \Call{ClassicalGreedy}{}. 
We define a multifamily of sets $\mathcal{F}^\downarrow$ as $\bigcup_{S\in \mathcal{F}}(2^S\setminus \emptyset)$ as the union of (nonempty) subsets of sets in $\mathcal{F}$, and we define the weight $w^\downarrow_A$ of set $A \in \mathcal{F}^\downarrow$ as the weight of its corresponding set in $\mathcal{F}$. 
In this context, we can interpret Line~\ref{line:setcoverknown_chooseS} in \Call{ClassicalGreedy}{} as the process of selecting an set $A$ from $\mathcal{F}^\downarrow$ that satisfies $A\subseteq R$ and minimizes $w^\downarrow_A/|A|$.

Let us describe our algorithm, \Call{LipschitzGreedy}{}, given in  Algorithm~\ref{alg:setcover_const}. 
It takes two positive integers, $K$ and $M$. 
Then, we sample two parameters $\bm{b}$ and $\bm{\pi}$ from uniform distributions over $[0,1]$ and permutations over $\{A \subseteq U\colon |A|\leq s\}$, respectively.
For the purpose we describe later, at Line~\ref{line:setcoverconst_is}, we round the weight of each set $S$ by setting $\bm{i}_S=\floor{M\log_s w_S+\bm{c}_1}$, where $\bm{c}_1$ is a constant determined by $K$, $M$, and $\bm{b}$, and the new weight $\widehat{\bm{w}}_S$ is given by $\bm{c}_2\cdot s^{\bm{i}_S/M} \approx w_S$, where $\bm{c}_2$ is a constant determined by $K$, $M$, and $\bm{b}$. 

At Lines~\ref{line:setcoverconst_add_condition} and~\ref{line:setcoverconst_add}, we extend the weight $\widehat{\bm{w}}_S$ onto $\bm{\mathcal{F}}^\downarrow$ described above.
Here, we prepare a triplet $(x,A,S)$ for each $A\in \bm{\mathcal{F}}^\downarrow$, where $x$ is the logarithm of the extended weight $\widehat{\bm{w}}^\downarrow_A$ and $S$ is the set in $\mathcal{F}$ corresponding to $A$. 
However, we do so only for $A\in \bm{\mathcal{F}}^\downarrow$ with $\bm{i}_S=-\infty$ or $\floor{\frac{\bm{i}_S-M\log_s |A|}{KM}} = \floor{\frac{\bm{i}_S}{KM}} =: \bm{l}(S)$, where the equality is guaranteed by Line~\ref{line:setcoverconst_add_condition}.
We call the value $\bm{l}(S)$ the \emph{level} of $S$. If $\bm{i}_S=-\infty$, we set $\bm{l}(S)=-\infty$.
We note that, for every $e \in S$, the tuple $(x,\{e\},S)$ is always added to $\bm{Q}$ because $M\log_s |\{e\}|=0$. 
Finally, in the loop starting at Line~\ref{line:setcoverconst_loop}, we execute the classical greedy algorithm on these sets and weights.

The approximation ratio is analyzed as follows. We know that the loop starting at Line~\ref{line:setcoverconst_loop} demonstrates $H_s$-approximation algorithm for the instance $(U,\bm{\mathcal{F}}^\downarrow,\widehat{\bm{w}}^\downarrow)$. 
To complete the analysis, we bound $\OPT_{U,\bm{\mathcal{F}}^\downarrow,\widehat{\bm{w}}^\downarrow}$ by $\OPT_{U,\mathcal{F}, w}$. First, we prove that the rounding process of the weight $w$ into $\widehat{\bm{w}}$ multiplies the approximation ratio by at most $s^{1/M}$.
Let $\mathcal{C}^*$ be the solution that attains $\OPT_{U,\mathcal{F}, w}$. 
We construct a solution $\widehat{\bm{\mathcal{C}}}^*$ for the instance $(U,\bm{\mathcal{F}}^\downarrow,\widehat{\bm{w}}^\downarrow)$ as follows: For each $S\in \mathcal{C}^*$, include it in $\widehat{\bm{\mathcal{C}}}^*$ if $(x,S,S)$ is added to $\bm{Q}$ in Line~\ref{line:setcoverconst_add}; otherwise, include all one-element subsets of $S$ in $\widehat{\bm{\mathcal{C}}}^*$.
The probability that each $S\in \mathcal{C}^*$ is not included in $\widehat{\bm{\mathcal{C}}}^*$ is at most $K^{-1}$ and in this case, at most $s$ one-element subsets of $S$ is included in $\widehat{\bm{\mathcal{C}}}^*$ instead. 
Therefore, $\OPT_{U,\mathcal{F}^\downarrow,\widehat{\bm{w}}^\downarrow}$ is bounded by $\left(1+\frac{s-1}{K}\right)$ times $\OPT_{U,\mathcal{F},\widehat{\bm{w}}}$, that is, $\left(1+\frac{s-1}{K}\right)\cdot s^{1/M}$ times $\OPT_{U,\mathcal{F},w}$, in expectation. By taking sufficiently large $K$ and $M$, we obtain an approximation guarantee of $H_s+\epsilon$.

Now we overview the analysis of the Lipschitz continuity.
Let us introduce an intuition to bound 
\begin{align*}
\EMW\left((\Call{LipschitzGreedy}{U,\mathcal{F},w,K,M},w), (\Call{LipschitzGreedy}{U,\mathcal{F},w+\delta \bm{1}_T, K, M}, w+\delta \bm{1}_T)\right).
\end{align*} 
For simplicity, consider the case where the level of $T$ is $l$ in both \Call{LipschitzGreedy}{$U,\mathcal{F},w,K,M$} and \Call{LipschitzGreedy}{$U,\mathcal{F},w+\delta \bm{1}_T,K,M$}. Before adding sets of level $l$, the algorithm's behavior is identical for both instances.
Similarly, after adding the sets of level $l$, $\bm{R}=\bigcup_{S\in \mathcal{F}, \bm{l}(S)\leq l}S$ is the same for both instances. Therefore, the behavior remains the same when adding sets of level $l+1$ or higher. 
Thus, only the sets with level $l$ affect the Lipschitz constant of the whole algorithm.

Here, we review the results of the dynamic distributed independent set by Censor-Hillel, Haramaty, and Karnin~\cite{censor2016optimal}. 
Although their original result is for the independent set problem on graphs, it can be easily interpreted for the set cover problem using a line graph.
For a set system $(U,\mathcal{F})$, a subfamily $\mathcal{P}\subseteq \mathcal{F}$ is called a \emph{set packing} if no two sets in $\mathcal{P}$ intersect. For a permutation $\pi$ over $\mathcal{F}$, we consider the following greedy algorithm \Call{Packing}{$U,\mathcal{F},\pi$} for the set packing problem: It iterates through the sets $S\in \mathcal{F}$ in order of $\pi(S)$ and includes a set $S$ in the solution only if it has no  intersection with any set included in the solution. 
This process allows us to construct a maximal set packing. The result in Censor-Hillel, Haramaty, and Karnin states the following. 
\begin{lemma}[\rm{\cite{censor2016optimal}}]\label{lem:randomizedgreedy}
If $\pi$ is chosen uniformly at random, then for any $X \in \mathcal{F}$, we have
\begin{align*}
\E_{\pi}\left[\Call{Packing}{U,\mathcal{F},\pi}\; \triangle\; \Call{Packing}{U,\mathcal{F}\setminus \{X\},\pi}\right]\leq 2,
\end{align*}
where we abuse notation $\pi$ to represent the permutation obtained by removing $\pi(X)$ from $\pi$, indicating the permutation on $\mathcal{F}\setminus X$.
\end{lemma}

For a fixed level $l$, there are at most $sKM$ different weights $\widehat{\bm{w}}^\downarrow(A)$ for sets $A \in \bm{\mathcal{F}}^\downarrow$ with level $l$, because there are at most $\log_{s^{1/M}}s^{K}=KM$ candidate values for $\bm{i}_S$ and $s$ values for $|A|$ (this is why we rounded the weight $w$ into $\widehat{\bm{w}}$).
Let $y_k$ be the $k$-th smallest value among such values and $\bm{\mathcal{C}}_k$ be the family of sets $A\in \bm{\mathcal{F}}^\downarrow$ with $\widehat{\bm{w}}^\downarrow_A=y_k$ that is added to $\bm{\mathcal{C}}$ in the loop starting at Line~\ref{line:setcoverconst_loop}.
Using Lemma~\ref{lem:randomizedgreedy}, we can bound the quantity $\E\left[\left|\bm{\mathcal{C}}_{k}\;\triangle\; \bm{\mathcal{C}}'_{k}\right|\right]$ using $\E\left[\left|\bm{\mathcal{C}}_{k'}\;\triangle\; \bm{\mathcal{C}}'_{k'}\right|\right]$ for $k' < k$ and the parameters $s$ and $f$. 
Consequently, $\E\left[\left|\bm{\mathcal{C}}_{k}\;\triangle\; \bm{\mathcal{C}}'_{k}\right|\right]$ can be bounded as a function of $s$, $f$ and $k$. Since at most $sKM$ different values are possible for $\widehat{\bm{w}}^\downarrow_A)$ with $l(A)=l$, we can bound the Lipschitz constant by a function of $s$, $f$, $K$ and $M$.

\subsection{Approximation Ratio}
Let us analyze the approximation ratio. 
We say that a set $S\in \mathcal{F}$ is \emph{hashed} (with respect to $\bm{i}_S$) if $\bm{i}_S \bmod KM < M\log_s |S|$. 
When a set $S$ is hashed, only singleton subsets of $S$ are added to the set $\bm{Q}$.
\begin{lemma}\label{lem:setcoverconst_approx}
\emph{\Call{LipschitzGreedy}{$U, \mathcal{F},w, K, M$}} has an approximation ratio $H_s\cdot \left(1+\frac{s-1}{K}\right)\cdot s^{1/M}$.
\end{lemma}
\begin{proof}
Let $\bm{\mathcal{F}}_{\mathrm{hash}} \subseteq \mathcal{F}$ be the set of hashed sets.
Let $\widehat{\bm{\mathcal{F}}}$ be the family of sets over $U$ defined by
\begin{align*}
\widehat{\bm{\mathcal{F}}}=(\mathcal{F}\setminus \bm{\mathcal{F}}_{\mathrm{hash}})\cup \{\{e\} : S \in \bm{\mathcal{F}}_{\mathrm{hash}}, e \in S \}.
\end{align*}
The weight $\widehat{\bm{w}}_S$ of a set $S \in \mathcal{F}\setminus \bm{\mathcal{F}}_{\mathrm{hash}}$ is defined as $s^{K \bm{b}+(\bm{i}_S+1)/M}$, and the weight $\widehat{\bm{w}}_{\{e\}}$ for $e \in S \in \bm{\mathcal{F}}_{\mathrm{hash}}$ is defined as $s^{K \bm{b}+(\bm{i}_S+1)/M}$.

Let $\mathcal{C}^*\subseteq \mathcal{F}$ be a minimum set cover of the instance $(U,\mathcal{F},w)$. Let $\widehat{\bm{\mathcal{C}}}^*$ be the family of sets in $\mathcal{F}$ defined by
\begin{align*}
(\mathcal{C}^*\setminus \bm{\mathcal{F}}_{\mathrm{hash}})\cup 
\{\{e\} : S\in \mathcal{C}^*\cap \bm{\mathcal{F}}_{\mathrm{hash}}, e \in S\}.
\end{align*}
It is clear that $\widehat{\bm{\mathcal{C}}}^*$ is a set cover of $(U, \widehat{\bm{\mathcal{F}}},\widehat{\bm{w}})$. Furthermore, we have  
\begin{align*}
& \E\left[\OPT_{U,\widehat{\bm{\mathcal{F}}},\widehat{\bm{w}}}\right]
\leq \E\left[\widehat{\bm{w}}\left(\widehat{\bm{\mathcal{C}}}^*\right)\right]
\leq s^{1/M}\E\left[w\left(\widehat{\bm{\mathcal{C}}}^*\right)\right] \\
&= s^{1/M}\left(\E\left[w(\mathcal{C}^*\setminus \bm{\mathcal{F}}_{\mathrm{hash}}) + s\cdot w(\mathcal{C}^*\cap \bm{\mathcal{F}}_{\mathrm{hash}})\right]\right)\\
&= s^{1/M}\left(w(\mathcal{C}^*) + (s-1)\E\left[w(\mathcal{C}^*\cap \bm{\mathcal{F}}_{\mathrm{hash}})\right]\right)\\
&= s^{1/M}\left(w(\mathcal{C}^*) + (s-1)\int_{0}^{1}\left(\sum_{S\in \mathcal{C}^*\colon \text{hashed w.r.t.\ }\bm{i}_S}w_S\right)\mathrm{d}b\right)\\
&\leq s^{1/M}\left(w(\mathcal{C}^*) + (s-1)\int_{0}^{1}\left(\sum_{S\in \mathcal{C}^*\colon \text{hashed w.r.t.\ } \bm{i}_S}w_S\right)\mathrm{d}b\right)\\
&= s^{1/M}\left(w(\mathcal{C}^*) + \frac{s-1}{K}w(\mathcal{C}^*)\right) 
\leq s^{1/M}\left(1+\frac{s-1}{K}\right)\OPT_{U,\mathcal{F},w}.
\end{align*}
Here, $\OPT_{U,\mathcal{F},w}$ and $\OPT_{U,\widehat{\bm{\mathcal{F}}},\widehat{\bm{w}}}$ represent the optimal values for $(U,\mathcal{F},w)$ and $(U,\widehat{\bm{\mathcal{F}}},\widehat{\bm{w}})$, respectively.

Recalling that \Call{LipschitzGreedy}{} simulates \Call{ClassicalGreedy}{$U,\widehat{\bm{\mathcal{F}}},\widehat{\bm{w}}$}, we have
\begin{align*}
\E\left[w\left(\Call{LipschitzGreedy}{U,\mathcal{F},w,K,M}\right)\right]\leq H_s\E\left[\OPT_{U,\widehat{\bm{\mathcal{F}}},\widehat{\bm{w}}}\right]\leq H_s\cdot s^{1/M}\left(1+\frac{s-1}{K}\right)\OPT_{U,\mathcal{F},w}.
\end{align*}
Therefore the lemma is proved.
\end{proof}

\subsection{Lipschitz Continuity}
Now we analyze the Lipschitz constant. 
By Lemma~\ref{lem:seeoneelement}, it suffices to evaluate
\begin{align*}
\EMW\left((\Call{LipschitzGreedy}{U,\mathcal{F},w,K,M}, w), (\Call{LipschitzGreedy}{U,\mathcal{F},w+\delta \bm{1}_T, K, M}, w+\delta\bm{1}_T)\right).
\end{align*}
Moreover, we can assume $0 < \delta\leq w_T$ or $w_T=0$.

Let $\bm{b}$ and $\bm{b}'$ be the parameter $\bm{b}$ sampled in
\Call{LipschitzGreedy}{} on $w$ and that on $w+\bm{1}_T$, respectively.
We use similar notations for $\bm{\pi}$, $\bm{i}_S$, $\bm{Q}$, $\bm{R}$, and $\bm{\mathcal{C}}$.
To bound the earth mover’s distance, we consider the coupling $\mathcal{W}$ between $\Call{LipschitzGreedy}{}$ on $w$ and that on $w+ \delta \bm{1}_T$
such that for each $b\in [0,1]$ and permutation $\pi$, the probability mass for $\bm{b} = b$ and $\bm{\pi}=\pi$ are transported to that for $\bm{b}' = b$ and $\bm{\pi}'=\pi$. 

Firstly, we evaluate the probability that $\bm{i}_T$ is not equal to $\bm{i}'_T$.
\begin{lemma}\label{lem:setcover_prob}
If $w_T>0$, we have $\Pr\left[\bm{i}_T\neq \bm{i}'_{T}\right]\leq \frac{\delta M}{w_T\log s}$.
\end{lemma}
\begin{proof}
We have
\begin{align*}
\Pr_{\mathcal{W}}\left[\bm{i}_T\neq \bm{i}'_{T}\right]
&= \Pr_{\bm{b}\sim \mathcal{U}[0,1]}\bigl[\floor{\left(\log_{s}w_T-K \bm{b}\right)M}\neq \floor{\left(\log_{s}(w_T+\delta)-K \bm{b}\right)M}\bigr]\\
&= \left(\log_s(w_T+\delta) - \log_s w_T\right)M
= \log_s\left(1+\frac{\delta}{w_T}\right)\cdot M
\leq \frac{\delta M}{w_T\log s},
\end{align*}
where the second equality is from $(K\bm{b}\bmod 1)\sim \mathcal{U}([0,1])$ because $K$ is an integer, and the inequality is from $\log (1+x)\leq x$.
\end{proof}
Consider the case $\bm{i}_T\neq \bm{i}'_T$.
For each $S\in \mathcal{F}$, let the \emph{level} of $S$ be $\bm{l}(S)=\floor{\frac{\bm{i}_S}{KM}}$ if $\bm{i}_S>-\infty$ and $\bm{l}(S)=-\infty$ if $\bm{i}_S=-\infty$. 
The following is easy but important.
\begin{lemma}\label{lem:setcover_incr}
For $(x_1,A_1,S_1), (x_2,A_2,S_2)\in \bm{Q}$, if $\bm{l}(S_1) < \bm{l}(S_2)$, then we have $x_1<x_2$.
\end{lemma}
\begin{proof}
$\bm{l}(S_2)>-\infty$ holds only if $w_{S_2}>0$. Therefore, the lemma is trivial if $w_{S_1}=0$. 
Assume $w_{S_1}>0$. Then, we have
\begin{align*}
x_1 < \left(\bm{l}(S_1)+1\right)K M
\leq \bm{l}(S_2)K M \leq x_2,
\end{align*}
where the last inequality is because in Line~\ref{line:setcoverconst_add} of \Call{LipschitzGreedy}{}, $(\bm{i}_{S_2}-M\log_s |A_2|, A_2, S_2)$ is added only if $\bm{i}_S\bmod KM\geq \bm{i}_{S_2}-M\log_s |A_2|$.
\end{proof}

From Lemma~\ref{lem:setcover_incr}, it can be observed that the loop starting from Line~\ref{line:setcoverconst_loop} in \Call{LipschitzGreedy}{} determines whether to include the set $S$ in the output $\bm{\mathcal{C}}$ in ascending order of their level $\bm{l}(S)$. 
Therefore, when all sets $S$ with $\bm{l}(S)\leq k$ have been examined, the set of elements covered by $\bm{\mathcal{C}}$ is equal to the set of elements belonging to at least one set with its level at most $k$.
Comparing the behavior of the loops starting from Line~\ref{line:setcoverconst_loop} in \Call{LipschitzGreedy}{$U,\mathcal{F},w,K,M$} and \Call{LipschitzGreedy}{$U,\mathcal{F},w+\delta \bm{1}_T,K,M$}, we have:
\begin{itemize}
    \item Up until the point where sets with levels less than $\bm{l}(T)$ are being examined, the behavior is exactly the same.
    \item Since the sets of elements covered by $\bm{\mathcal{C}}$ and $\bm{\mathcal{C}}'$ are the same just after all sets with level $\bm{l}'(T)$ has been examined (because all one-element subsets of a set $S \in \mathcal{F}$ with $\bm{l}(S) \leq \bm{l}'(T)$ have the same level as $S$), the algorithm behaves the same when examining sets with a level greater than $\bm{l}'(T)$.
\end{itemize}
We can assume without loss of generality that there exists no set $S \in \mathcal{F}$ such that $\bm{l}(T) < \bm{l}(S) < \bm{l}'(T)$. 
In particular, any set that differs between $\bm{\mathcal{C}}$ and $\bm{\mathcal{C}}'$ has either level $\bm{l}(T)$ or $\bm{l}'(T)$. 
Let $\bm{X}$ be the set of all possible values taken by the first component $x$ of the tuples $(x,A,S)$ in $\bm{Q}$ such that $\bm{l}(S)=\bm{l}(T)$ or $\bm{l}(S)=\bm{l}'(T)$. In other words, let
\begin{align*}
\bm{X}=&\bigcup_{i=1}^{s}\bigcup_{x=\ceil{M\log_s i}}^{K M - 1}\left\{\bm{l}(T)K M + x -M \log_s i, \bm{l}'(T)K M + x -M \log_s i\right\}
\end{align*}
when $w_T>0$ and 
\begin{align*}
\bm{X}=\{-\infty\}\cup \bigcup_{i=1}^{s}\bigcup_{x=\ceil{M\log_s i}}^{K M - 1}\left\{\bm{l}'(T)K M + x -M \log_s i\right\}
\end{align*}
when $w_T=0$.

For $k=1,\dots, |\bm{X}|$, let $\bm{y}_k$ be the $k$-th smallest number in $\bm{X}$ and let $\bm{\mathcal{C}}_k$ (resp., $\bm{\mathcal{C}}'_k$) be the family of sets $S \in \mathcal{F}$ that was added to $\bm{\mathcal{C}}$ when having examined a tuple $(x,A,S) \in \bm{Q}$ with $x=\bm{y}_k$ in \Call{LipschitzGreedy}{$U,\mathcal{F},w,K,M$} (resp., \Call{LipschitzGreedy}{$U,\mathcal{F},w+\delta \bm{1}_T,K,M$}).

\begin{lemma}\label{lem:setcoverconst_pack}
Let $k\in \{1,\dots, |X|\}$. Then, we have
\begin{align*}
\E_{\mathcal{W}}\left[\left|\bm{\mathcal{C}}_k\; \triangle\; \bm{\mathcal{C}}'_k\right|\right]\leq (2^{s}sf+1)^{k+3}.\qedhere
\end{align*}
\end{lemma}
\begin{proof}
We prove this by induction on $k$.
Let $\bm{D}_k$ (resp. $\bm{D}'_k$) be the set of tuples $(x,A,S)$ in $\bm{Q}$ with $x=y_k$ such that $A \subseteq \bm{R}$ holds when \Call{LipschitzGreedy}{$U,\mathcal{F},w,K,M$} (resp., \Call{LipschitzGreedy}{$U,\mathcal{F},w+\delta \bm{1}_T,K,M$}) starts to examine tuples $(\tilde{x},\tilde{A},\tilde{S})$ in $\bm{Q}$ with $\tilde{x} = y_k$. 
Furthermore, let $\bm{R}_k$ (resp. $\bm{R}'_k$) be $\bm{R}$ at that time.
From Lemma~\ref{lem:randomizedgreedy}, we have
\begin{align*}
\E_{\mathcal{W}}\left[\left|\bm{\mathcal{C}}_k\;\triangle\; \bm{\mathcal{C}}'_k\right|\right]\leq 2\cdot \E_{\mathcal{W}}\left[\left|\bm{D}_k\;\triangle\; \bm{D}'_k\right|\right].
\end{align*}
Moreover,
\begin{align*}
\E_{\mathcal{W}}\left[\left|\bm{R}_{k+1}\;\triangle\; \bm{R}'_{k+1}\right|\right]\leq s\cdot \E_{\mathcal{W}}\left[\sum_{i=1}^{k}\left|\bm{\mathcal{C}}_k\;\triangle\; \bm{\mathcal{C}}'_k\right|\right].
\end{align*}
Since each $e \in \bm{R}_{k+1}\;\triangle\; \bm{R}'_{k+1}$ belongs to $A$ for at most $2^sf$ tuples $(x,A,S)\in Q$, we have
\begin{align*}
\E_{\mathcal{W}}\left[\left|\bm{D}_{k+1}\;\triangle\; \bm{D}'_{k+1}\right|\right]\leq f\cdot 2^{s}\cdot \E_{\mathcal{W}}\left[\left|\bm{R}_{k+1}\;\triangle\; \bm{R}'_{k+1}\right|\right]+2\cdot 2^s,
\end{align*}
where the last term counts the number of tuples $(x,A,S)$ with $S=T$.
Thus, we have
\begin{align*}
& \E_{\mathcal{W}}\left[\left|\bm{\mathcal{C}}_{k+1}\;\triangle\; \bm{\mathcal{C}}'_{k+1}\right|\right]
\leq 2\cdot \E_{\mathcal{W}}\left[\left|\bm{D}_{k+1}\;\triangle\; \bm{D}'_{k+1}\right|\right]\\
&\leq f\cdot 2^s\cdot \E_{\mathcal{W}}\left[\left|\bm{R}_{k+1}\;\triangle\; \bm{R}'_{k+1}\right|\right]+2\cdot 2^s
\leq f\cdot 2^s \cdot s\cdot \E_{\mathcal{W}}\left[\sum_{i=1}^{k}\left|\bm{\mathcal{C}}_k\;\triangle\; \bm{\mathcal{C}}'_k\right|\right]+2\cdot 2^s.
\end{align*}
Thus, we obtain
\begin{align*}
\E_{\mathcal{W}}\left[\left|\bm{\mathcal{C}}_k\;\triangle\; \bm{\mathcal{C}}'_k\right|\right]\leq \left(2^{s}sf+1\right)^{k+3}-2\cdot 2^s
\end{align*}
by induction. Therefore the lemma holds.
\end{proof}

Since there are at most $2sKM$ numbers in $X$, we have the following.
\begin{lemma}\label{lem:setcover_num}
We have
\begin{align*}
\E_{\mathcal{W}}\left[\left|\bm{\mathcal{C}}\;\triangle\; \bm{\mathcal{C}}'\right|\right]\leq \left(2^{s}sf+1\right)^{2sK M+4}.
\end{align*}
\end{lemma}
\begin{proof}
We have
\[
\E_{\mathcal{W}}\left[\left|\bm{\mathcal{C}}\;\triangle\; \bm{\mathcal{C}}'\right|\right]
= \sum_{k=1}^{|X|}\E_{\mathcal{W}}\left[\left|\bm{\mathcal{C}}_k\;\triangle\; \bm{\mathcal{C}}'_k\right|\right]
\leq \sum_{k=1}^{|X|}(2^{s}sf+1)^{k+3}\leq (2^{s}sf+1)^{|X|+4}\leq \left(2^{s}sf+1\right)^{2sK M+4}.\qedhere
\]
\end{proof}

The following is straightforward.
\begin{lemma}\label{lem:setcover_len}
Let $S\in \mathcal{F}$ and assume $S$ has level $\bm{l}(T)$ or $\bm{l}'(T)$.
Then, $w_S\leq s^{K}(w_T+\delta)$.
\end{lemma}
\begin{proof}
Since $\bm{l}(S)\leq \bm{l}'(T)$, we have
\begin{align*}
w_S&\leq s^{K b+(i_S+1)/M}\leq s^{K b+(KM\cdot \bm{l}(S) + KM-1+1)/M}
\leq s^{K}\cdot s^{K b + (KM \cdot \bm{l}(S))/M}\\
&\leq s^{K}\cdot s^{K b + (KM \cdot \bm{l}'(T))/M}\leq s^{K}\cdot w'(T)=s^{K}(w_T+\delta).\qedhere
\end{align*}
\end{proof}

Therefore, the following holds.
\begin{lemma}
We have
\begin{align*}
&\EMW\left(\left(\Call{LipschitzGreedy}{U,\mathcal{F},w,K,M},w\right), \left(\Call{LipschitzGreedy}{U,\mathcal{F},w+\delta \bm{1}_T, K, M},w+\delta \bm{1}_T\right)\right)\\
&\leq 3M \left(2^{s}sf+1\right)^{5sK M}\cdot \delta.
\end{align*}
\end{lemma}
\begin{proof}
If $w_T=0$, we have 
\begin{align*}
&\EMW\left(\left(\Call{LipschitzGreedy}{U,\mathcal{F},w,K,M},w\right), \left(\Call{LipschitzGreedy}{U,\mathcal{F},w+\delta \bm{1}_T, K, M},w+\delta \bm{1}_T\right)\right)\\
&\leq \E\left[\left|\bm{\mathcal{C}}\;\triangle\; \bm{\mathcal{C}}'\right|\right]\cdot s^{K}(0+\delta)\leq \left(2^{s}sf+1\right)^{2sK M+4} \cdot s^{K}\cdot \delta 
\leq \left(2^{s}sf+1\right)^{2sK M+4+K}\cdot \delta
\leq \left(2^{s}sf+1\right)^{5sK M}\cdot \delta,
\end{align*}
where the second inequality is from Lemmas~\ref{lem:setcover_num} and~\ref{lem:setcover_len}, and the last inequality is because $M \geq 1$.
Otherwise, we have
\begin{align*}
&\EMW\left(\left(\Call{LipschitzGreedy}{U,\mathcal{F},w,K,M},w\right), \left(\Call{LipschitzGreedy}{U,\mathcal{F},w+\delta \bm{1}_T, K, M},w+\delta \bm{1}_T\right)\right)\\
&\leq \Pr\left[\bm{i}_T\neq \bm{i}'_T\right]\cdot \E\left[\left|\bm{\mathcal{C}}\;\triangle\; \bm{\mathcal{C}}'\right|\right]\cdot s^{K}(w_T+\delta)\\
&\leq \frac{\delta M}{w_T\log s}\cdot \left(2^{s}sf+1\right)^{2sK M+4}\cdot s^{K}\cdot 2w_T
\leq 3M \left(2^{s}sf+1\right)^{2sK M+4+K}\cdot \delta
\leq 3M \left(2^{s}sf+1\right)^{5sK M}\cdot \delta,
\end{align*}
where the second inequality is from Lemmas~\ref{lem:setcover_prob},~\ref{lem:setcover_num},~\ref{lem:setcover_len} and $\delta\leq w_T$, and the third inequality is from $\log s \geq \log 2\geq \frac{2}{3}$.
\end{proof}

Now we have the following.
\begin{lemma}
\Call{LipschitzGreedy}{} has Lipschitz constant $3M \left(2^{s}sf+1\right)^{5sK M}$.
\end{lemma}

Next, we prove the main theorem.
\begin{proof}[Proof of Theorem~\ref{thm:set-cover-greedy}]
Let $K=\ceil{4\epsilon^{-1} H_s}(s-1)$ and $M=\ceil{4\epsilon^{-1} H_s\log s}$. Then, from Lemma~\ref{lem:setcoverconst_approx}, the approximation ratio is
\begin{align*}
& H_s\cdot \left(1+\frac{s-1}{K}\right)\cdot s^{1/M}
\leq H_s \cdot \left(1+\frac{\epsilon}{4H_s}\right)\cdot 2^{\frac{\epsilon}{4H_s}}\\
&\leq H_s \cdot \left(1+\frac{\epsilon}{4H_s}\right)\cdot \left(1+\frac{\epsilon}{4H_s}\right)
\leq H_s \cdot \left(1+\frac{\epsilon}{H_s}\right)
= H_s+\epsilon,
\end{align*}
where the second inequality is from $2^{x}\leq 1+x$ holds for $0\leq x \leq 1$, and the third inequality is from $\epsilon\leq 1$. In this setting of $K$ and $M$, the Lipschitz constant of \Call{LipschitzGreedy}{} is at most
\begin{align*}
& 3M \left(2^{s}sf+1\right)^{5sK M} 
\leq 3\cdot 5\epsilon^{-1} H_s\log s \cdot \left(2^{s}sf+1\right)^{5s\cdot \left(5\epsilon^{-1} H_s\right)^2\cdot (s-1)\log s}\\
&\leq 2^{O\left(\log H_s + \log \epsilon^{-1} + \log \log s + (s+\log s+\log f)\cdot s\cdot H_s^2\cdot \epsilon^{-2}\cdot s\log s\right)}
\leq 2^{O\left(\epsilon^{-2}\left(s+\log f\right)s^2\log^3 s\right)},
\end{align*}
where the first inequality is from $K, M\leq 5\epsilon^{-1}$ follows from $\epsilon \leq 1$.
\end{proof}

%% file: set-cover-lp.tex

\section{LP-based $O(\log n)$-Approximation Algorithm for Set Cover}\label{sec:set-cover}

In this section, we design a Lipschitz-continuous algorithm for the minimum set cover problem that has a sublinear Lipschitz constant for any instance.
Specifically, we show the following:
\begin{theorem}\label{thm:set-cover}
    There exists a polynomial-time randomized algorithm for the minimum set cover problem that, given a set cover instance $(U,\mathcal{F})$ on $n$ elements and $m$ sets and a weight vector $w \in \mathbb{R}_{\geq 0}^{\mathcal{F}}$, outputs a family of sets $\bm{\mathcal{S}} \subseteq \mathcal{F}$ with the following properties:
    \begin{itemize}
    \item $\bm{\mathcal{S}}$ is a feasible solution with probability $1-1/\mathrm{poly}(n)$,
    \item the total weight of $\bm{\mathcal{S}}$ satisfies $\E[\sum_{S \in \bm{\mathcal{S}}}w_S] = O(\log n \cdot (\OPT+c))$, where $c>0$ is an arbitrarily small constant, and
    \item the algorithm has a Lipschitz constant $O(\sqrt{\min\{n,m\}\log m} \log n)$.
    \end{itemize}
\end{theorem}


Our algorithm first solves a (regularized) LP relaxation and then rounds the obtained solution so that the whole process is Lipschitz continuous.
These two steps are explained and analyzed in Sections~\ref{subsec:set-cover-lp} and~\ref{subsec:set-cover-rounding}, respectively.
Then, we provide the proof of Theorem~\ref{thm:set-cover} in Section~\ref{subsec:set-cover-whole}.

\subsection{LP Relaxation}\label{subsec:set-cover-lp}

We consider the following natural LP relaxation for the set cover problem:
\begin{align}
    \begin{array}{lll}
        \text{minimize} & \displaystyle \sum_{S \in \mathcal{F}}w_S x_S, \\
        \text{subject to} & \displaystyle \sum_{S \in \mathcal{F}: e \in S}x_S \geq 1 & \forall e \in U, \\
        & \displaystyle \sum_{S \in \mathcal{F}} x_S \leq n, \\
        & 0 \leq x_S \leq 1 & \forall S \in \mathcal{F}.
    \end{array}\label{eq:set-cover-lp}
\end{align}
Note that the second constraint is valid because there is always an optimal solution that uses at most $n$ elements.
However, the solution to this LP relaxation may not be Lipschitz continuous with respect to the weight vector $w \in \mathbb{R}_{\geq 0}^{\mathcal{F}}$.
Hence, we consider the following optimization problem with additional regularizers and constraints.
\begin{align}
    \begin{array}{lll}
        \mathrm{minimize} & \displaystyle \sum_{S \in \mathcal{F}} w_S x_S + \frac{\lambda}{2} \sum_{S \in \mathcal{F}}x_S^2 + \frac{\kappa}{2} \sum_{S \in \mathcal{F}} \left(\frac{w_S x_S}{2W}+\frac{1}{2m}\right) \log \left(\frac{w_S x_S}{2W} + \frac{1}{2m}\right)\\
        \text{subject to} & \displaystyle \sum_{S \in \mathcal{F}: e \in S}x_S \geq 1 & \forall e \in U, \\
        & \displaystyle \sum_{S \in \mathcal{F}} x_S \leq n, \\        
        & 0 \leq x_S \leq 1 & \forall S \in \mathcal{F},
    \end{array}\label{eq:set-cover-regularized-lp}
\end{align}
where $\lambda,\kappa,W>0$ are parameters.
The first regularizer is simply a squared $\ell_2$ norm of $x$.
The second regularizer is the negative entropy of the vector $(w_S x_S/(2W)+1/(2m))_{S \in \mathcal{F}}$.
We will choose $W \approx \OPTLP$, where $\OPTLP$ is the optimal value of LP~\eqref{eq:set-cover-lp}, so that $\sum_{S \in \mathcal{F}}(w_S x_S/2W + 1/(2m)) \leq 1$, and hence the negative entropy term contributes by $\log m$.

Let \Call{RegularizedLP}{} denote the algorithm that takes an instance $(U,\mathcal{F},w)$ of the set cover problem and three parameters $\lambda$, $\kappa$, and $W$, and then outputs the solution to~\eqref{eq:set-cover-regularized-lp}.
In the rest of this section, we show the following:
\begin{lemma}\label{lem:set-cover-regularized-lp}
    Suppose $W \geq \OPTLP$. 
    Then, \Call{RegularizedLP}{} returns a feasible solution $x \in \mathbb{R}_{\geq 0}^\mathcal{F}$ to LP~\eqref{eq:set-cover-lp} such that
    \[
        \sum_{S \in \mathcal{F}}w_S x_S \leq \OPTLP + \frac{\lambda \min\{n,m\}}{2} + \frac{\kappa \log m}{2}.
    \]
    For $w,w' \in \mathbb{R}_{\geq 0}^{\mathcal{F}}$, let $x = \Call{RegularizedLP}{U,\mathcal{F},w,\lambda,\kappa,W}$ and $x' = \Call{RegularizedLP}{U,\mathcal{F},w',\lambda,\kappa,W}$.
    If $\max\{w^\top x, w^\top x', (w')^\top x, (w')^\top x' \} \leq W$, then we have
    \[
        \sum_{S \in \mathcal{F}}|w_S x_S - w'_S x'_S| 
        = O\left(\delta (\kappa + W)\sqrt{1+\frac{1}{\kappa \lambda}}\right) \cdot \|w-w'\|_1.
    \]
\end{lemma}


\begin{proof}[Proof of the first part of Lemma~\ref{lem:set-cover-regularized-lp}]
    Let $x^* \in \mathbb{R}^\mathcal{F}$ be the optimal solution to LP~\eqref{eq:set-cover-lp}.
    Then, the objective function is bounded as
    \begin{align*}
        & \sum_{S \in \mathcal{F}} w_S x_S \\
        & \leq 
        \sum_{S \in \mathcal{F}} w_S x_S + \frac{\lambda}{2} \sum_{S \in \mathcal{F}}x_S^2 + \frac{\kappa}{2} \sum_{S \in \mathcal{F}} \left(\frac{w_S x_S}{2W}+\frac{1}{2m}\right) \log \left(\frac{w_S x_S}{2W} + \frac{1}{2m}\right) + \frac{\kappa \log m}{2} \\
        & \leq \sum_{S \in \mathcal{F}} w_S x^*_S + \frac{\lambda}{2} \sum_{S \in \mathcal{F}}(x^*_S)^2 + \frac{\kappa}{2} \sum_{S \in \mathcal{F}} \left(\frac{w_S x^*_S}{2W}+\frac{1}{2m}\right) \log \left(\frac{w_S x^*_S}{2W} + \frac{1}{2m}\right) + \frac{\kappa \log m}{2} \\
        & \leq \OPTLP + \frac{\lambda \min\{n,m\}}{2} + \frac{\kappa \log m}{2}.
        \qedhere
    \end{align*}
\end{proof}
Henceforth, we focus on analyzing the Lipschitz constant of \Call{RegularizedLP}{}.
Fix $\kappa,\lambda,W >0$, and let $f_w: \mathbb{R}_{\geq 0}^\mathcal{F} \to \mathbb{R}$ denote the objective function of~\eqref{eq:set-cover-regularized-lp}, parameterized by the weight vector $w \in \mathbb{R}_{\geq 0}^{\mathcal{F}}$.

Fix a weight vector $w \in \mathbb{R}_{\geq 0}^{\mathcal{F}}$ and a set $S \in \mathcal{F}$.
Let $w'\in \mathbb{R}_{\geq 0}^{\mathcal{F}}$ be a vector obtained from $w$ by increasing $w_S$ by $\delta > 0$.
By a similar argument to the proof of Lemma~\ref{lem:seeoneelement}, it suffices to show that the claim holds for these particular $w$ and $w'$.
Let $x$ and $x'$ be the optimal solutions to~\eqref{eq:set-cover-regularized-lp} for the instances $(U,\mathcal{F},w)$ and $(U,\mathcal{F},w')$, respectively.
We consider the following quantity:
\[
    D := f_{w}(x') - f_{w}(x) + f_{w'}(x) - f_{w'}(x') \geq 0,
\]
where the inequality holds because $x$ minimizes $f_w$ and $x'$ minimizes $f_{w'}$.

For notational simplicity, for a vector $x \in \mathbb{R}^{\mathcal{F}}$, define $H(x) = \sum_{S \in \mathcal{F}}x_S \log x_S$ as the negative entropy function.
Also for two vectors $x,y \in \mathbb{R}^{\mathcal{F}}$, let $x \circ y$ be their Hadamard product, i.e., $(x\circ y)_S = x_S y_S$ for any $S \in \mathcal{F}$.
The following gives an upper bound on $D$:
\begin{lemma}\label{lem:D-upper-bound}
    We have
    \[
        D \leq \delta (x_S - x'_S) \left( 1 + \frac{\kappa}{2W} \right).
    \]
    In particular, $x_S \geq x'_S$ holds because $D \geq 0$.
\end{lemma}
\begin{proof}
    We first note that 
    \begin{align}
        & D = f_w(x') - f_w(x) + f_{w'}(x) - f_{w'}(x') \nonumber \\
        & = w_S x'_S - w_S x_S + w'_S x_S - w'_S x'_S \nonumber \\
        & \qquad + \frac{\kappa}{2} \left( H\left(\frac{w \circ x'}{2W} + \frac{\bm{1}}{2m}\right) - H\left(\frac{w \circ x}{2W} + \frac{\bm{1}}{2m}\right) + H\left(\frac{w' \circ x}{2W} + \frac{\bm{1}}{2m}\right) - H\left(\frac{w' \circ x'}{2W} + \frac{\bm{1}}{2m}\right)\right) \nonumber \\
        & = \delta (x_S - x'_S) \nonumber \\
        & \qquad + \frac{\kappa}{2} \left( H\left(\frac{w \circ x'}{2W} + \frac{\bm{1}}{2m}\right) - H\left(\frac{w \circ x}{2W} + \frac{\bm{1}}{2m}\right) + H\left(\frac{w' \circ x}{2W} + \frac{\bm{1}}{2m}\right) - H\left(\frac{w' \circ x'}{2W} + \frac{\bm{1}}{2m}\right)\right)
        \label{eq:D-upper-bound}
    \end{align}

    Now, we evaluate the second term.
    Let $h(\hat{w}_S,\hat{x}) = H((w - w_S \bm{1}_S + \hat{w}_S \bm{1}_S)\circ \hat{x}/(2W) +\bm{1}/(2m))$.
    Then, by applying the mean-value theorem twice, we can show that there exist $w_S \leq \tilde{w}_S \leq w'_S$ and a convex combination $\tilde{x} \in \mathbb{R}^\mathcal{F}$ of $x$ and $x'$ such that 
    \begin{align*}
        & 
        H\left(\frac{w \circ x'}{2W} + \frac{\bm{1}}{2m}\right) - H\left(\frac{w \circ x}{2W} + \frac{\bm{1}}{2m}\right) + H\left(\frac{w' \circ x}{2W} + \frac{\bm{1}}{2m}\right) - H\left(\frac{w' \circ x'}{2W} + \frac{\bm{1}}{2m}\right) \\
        & = h(w_S,x')- h(w_S,x) + h(w'_S,x) - h(w'_S,x') \\
        & = (w'_S - w_S) \frac{\partial}{\partial \hat{w}_S} (h(\cdot,x) - h(\cdot,x'))(\tilde{w}_S) \\
        & = (w'_S - w_S) \cdot (x-x')^\top \frac{\partial^2 h}{\partial \hat{w}_S \partial \hat{x}}(\tilde{w}_S,\tilde{x}) \\
        & = (w'_S - w_S)(x_S-x'_S) \frac{1}{2W} \left(2 - \frac{W}{W+m\tilde{w}_S\tilde{x}} + \log\left(\frac{1}{2m} + \frac{\tilde{w}_S \tilde{x}_S}{2 W}\right) \right) \\
        & = \frac{\delta}{2W}(x_S-x'_S) \left(2 - \frac{ W}{ W + m \tilde{w}_S \tilde{x}_S} + \log\left(\frac{1}{2m} + \frac{\tilde{w}_S \tilde{x}_S}{2 W}\right) \right).  
    \end{align*}
    The third factor is an increasing function in $\tilde{w}_S \tilde{x}_S$, and hence it can be bounded from above by replacing $\tilde{w}_S \tilde{x}_S$ with $W$.
    Then, we have
    \begin{align*}
        & 
        H\left(\frac{w \circ x'}{W} + \frac{\bm{1}}{2m}\right) - H\left(\frac{w \circ x}{W} + \frac{\bm{1}}{2m}\right) + H\left(\frac{w' \circ x}{W} + \frac{\bm{1}}{2m}\right) - H\left(\frac{w' \circ x'}{W} + \frac{\bm{1}}{2m}\right) \\
        & \leq \frac{\delta}{2W}(x_S-x'_S) \left(2 - \frac{1}{1 + m} + \log\frac{1 + m}{2 m}\right) \\
        & \leq \frac{\delta (x_S-x'_S)}{W}.
    \end{align*}
    Combining this with~\eqref{eq:D-upper-bound}, the the claim holds.
\end{proof}

Now using this upper bound, we analyze the difference between $x$ and $x'$.
\begin{lemma}\label{lem:x-difference}
    We have
    \[
        x_S - x'_S \leq \frac{\delta}{\lambda} \left( 1 + \frac{\kappa}{2W} \right).
    \]
\end{lemma}
\begin{proof}
    As $x \mapsto \|x\|_2^2/2$ is $1$-strongly convex with respect to $\|\cdot\|_2$ and so is the objective function of~\eqref{eq:set-cover-regularized-lp}, we have
    \begin{align*}
        f_w(x') - f_w(x)
        & \geq \langle \nabla f_w(x), x'-x \rangle + \frac{\lambda}{2} \|x' - x\|_2^2
        \geq \frac{\lambda}{2} \|x' - x\|_2^2, \text{ and}\\
        f_{w'}(x) - f_{w'}(x')
        & \geq
        \langle \nabla f_{w'}(x), x - x' \rangle + \frac{\lambda}{2} \|x - x'\|_2^2 \geq \frac{\lambda}{2} \|x' - x\|_2^2.
    \end{align*}
    Hence, we have
    \begin{align}
        D \geq \lambda \|x - x'\|_2^2
        \label{eq:d-lower-bound-1}
    \end{align}

    Combining Lemma~\ref{lem:D-upper-bound} and~\eqref{eq:d-lower-bound-1}, we obtain
    \begin{align*}
        \|x' - x\|_2^2 \leq \frac{\delta (x_S - x'_S)}{\lambda} \left( 1 + \frac{\kappa}{2W} \right).
    \end{align*}
    In particular, this implies
    \[
        x_S - x'_S \leq \frac{\delta}{\lambda} \left( 1 + \frac{\kappa}{2W} \right). 
        \qedhere
    \]
\end{proof}
Combining Lemmas~\ref{lem:D-upper-bound} and~\ref{lem:x-difference}, the following holds.
\begin{corollary}\label{cor:D-upper-bound}
    We have
    \[
        D \leq \frac{\delta^2}{\lambda} \left( 1 + \frac{\kappa}{2W} \right)^2.
    \]
\end{corollary}

\begin{proof}[Proof of the second part of Lemma~\ref{lem:set-cover-regularized-lp}]
    Note that  $x \mapsto H(x)$ is $1$-strongly convex with respect to $\|\cdot\|_1$, given that $x \geq 0$ and $\|x\|\leq 1$.
    Then, using the fact that $\max\{w^\top x, w^\top x', (w')^\top x, (w')^\top x' \} \leq W$, we have the following:
    \begin{align*}
        & f_w(x') - f_w(x) 
        \geq \langle \nabla f_w(x), x'-x \rangle + \frac{\kappa}{2} \left\|\left(\frac{w \circ x'}{W} + \frac{\bm{1}}{2m}\right) - \left(\frac{w\circ x}{W} + \frac{\bm{1}}{2m}\right)\right\|_1^2 \\
        & \geq \frac{\kappa}{2} \left\|\frac{w \circ x'}{W} - \frac{w \circ x}{W}\right\|_1^2 = \frac{\kappa}{2W^2} \left\|w \circ (x' - x)\right\|_1^2  , \text{ and}\\
        & f_{w'}(x) - f_{w'}(x')
        \geq
        \langle \nabla f_{w'}(x), x - x' \rangle + \frac{\kappa}{2} \left\|\left(\frac{w'\circ x}{W} + \frac{\bm{1}}{2m}\right) - \left(\frac{w'\circ x'}{W} + \frac{\bm{1}}{2m}\right)\right\|_1^2 \\
        & \geq \frac{\kappa}{2} \left\|\frac{w' \circ x'}{W} - \frac{w' \circ x}{W}\right\|_1^2
        = \frac{\kappa}{2W^2} \left\|w' \circ (x' - x)\right\|_1^2.
    \end{align*}
    Hence, we have
    \begin{align}
        & D \geq \frac{\kappa}{2W^2} \left(\left\|w \circ (x' - x)\right\|_1^2 + \left\|w' \circ (x' - x)\right\|_1^2\right) \nonumber \\
        & = \frac{\kappa}{2W^2} \left\|w' \circ x' - w \circ x - \delta x'_S \bm{1}_S \right\|_1^2 + \frac{\kappa}{2W^2} \left\|w' \circ x' - w \circ x + \delta x_S \bm{1}_S\right\|_1^2 \nonumber \\
        & \geq \frac{\kappa}{2W^2} \left(\|w' \circ x' - w \circ x\|_1^2 - 2\|w' \circ x' - w \circ x\|_1\| \delta x'_S \bm{1}_S \|_1 +  \|\delta x'_S \bm{1}_S \|_1^2   \right) \nonumber \\
        & \qquad + \frac{\kappa}{2W^2}\left(\|w' \circ x' - w \circ x\|_1^2 - 2 \|w' \circ x' - w \circ x\|_1 \|\delta x_S \bm{1}_S\|_1  + \|\delta x_S \bm{1}_S\|_1^2\right) \tag{by the inverse triangle inequality} \\
        & \geq \frac{\kappa}{W^2}\|w' \circ x' - w \circ x\|_1^2 -  \frac{2 \delta \kappa}{W^2}\|w' \circ x' - w \circ x\|_1. \label{eq:d-lower-bound-2}
    \end{align}

    Combining Corollary~\ref{cor:D-upper-bound} and~\eqref{eq:d-lower-bound-2}, we have
    \[
        \frac{\kappa}{W^2}\|w' \circ x' - w \circ x\|_1^2 -  \frac{2 \delta \kappa}{W^2}\|w' \circ x' - w \circ x\|_1 \leq \frac{\delta^2}{\lambda} \left( 1 + \frac{\kappa}{2W} \right)^2
    \]
    Solving this inequality for $\|w' \circ x' - w \circ x\|_1^2$, we obtain
    \begin{align*}
        & \|w' \circ x' - w \circ x\|_1 
        \leq 
        \delta + 
        \frac{\delta}{2} \sqrt{\frac{(2 W + \kappa)^2 + 
         4 \kappa \lambda}{\kappa \lambda}}
        =  
        O\left(\delta (\kappa + W)\sqrt{1+\frac{1}{\kappa \lambda}}\right)
    \end{align*}
    Hence, the claim holds.    
\end{proof}


\subsection{Rounding}\label{subsec:set-cover-rounding}

\begin{algorithm}[t!]
    \caption{Rounding algorithm for the set cover problem}\label{alg:set-cover-rounding}
    \Procedure{\emph{\Call{RoundingSC}{$U,\mathcal{F},x$}}}{
        \KwIn{A set $U$, $\mathcal{F} \subseteq 2^U$, and a vector $x \in [0,1]^{\mathcal{F}}$.}
        $\bm{\mathcal{S}} \gets \emptyset$\;
        $K \gets \Theta(\log n)$\;
        \For{$K$ times}{
            \For{$S \in \mathcal{F}$}{
                Add $S$ to $\bm{\mathcal{S}}$ with probability $x_S$.
            }
        }
        \Return $\bm{\mathcal{S}}$.
    }
\end{algorithm}
Next, we design a Lipschitz continuous rounding algorithm for LP~\eqref{eq:set-cover-lp}.
Our rounding algorithm, \Call{RoundingSC}{}, given in Algorithm~\ref{alg:set-cover-rounding} is standard, and our focus is to show that it has a stability property that can help show the Lipschitz continuity of our algorithm for the set cover problem.
\begin{lemma}\label{lem:set-cover-rounding}
    Let $(U,\mathcal{F})$ be a set cover instance on $n$ elements and $w \in \mathbb{R}_{\geq 0}^{\mathcal{F}}$ be a weight vector.
    Let $\bm{\mathcal{S}} = \Call{RoundingSC}{U,\mathcal{F},x}$, where $x \in \mathbb{R}^{\mathcal{F}}$ is a feasible solution to LP~\eqref{eq:set-cover-lp}.
    Then, $\bm{\mathcal{S}}$ is a feasible set cover with probability $1-1/\mathrm{poly}(n)$, and
    \[
        \E\left[\sum_{S \in \bm{\mathcal{S}}}w_S\right] = O\left(\log n \cdot \sum_{S \in \mathcal{F}}w_S x_S \right).
    \]

    Let $w,w' \in \mathbb{R}_{\geq 0}^{\mathcal{F}}$ be weight vectors and     $x,x' \in [0,1]^{\mathcal{F}}$ be feasible solutions to LP~\eqref{eq:set-cover-lp}.
    For $\bm{\mathcal{S}} = \Call{RoundingSC}{U,\mathcal{F},x}$ and $\bm{\mathcal{S}}' = \Call{RoundingSC}{U,\mathcal{F},x'}$, we have   \[
        \EMW((\bm{\mathcal{S}},w),(\bm{\mathcal{S}}',w')) = O\left(\left( \|w-w'\|_1 +\sum_{S \in \mathcal{F}} |w_S x_S  - w'_S x'_S |\right) \log n\right).
    \]
\end{lemma}
\begin{proof}
    We first show that $\bm{\mathcal{S}}$ is feasible with high probability.
    The probability that an element $e \in U$ is not covered is
    \begin{align*}
        & \prod_{S \in \mathcal{F} : e \in S} (1 - x_S)^{K}
        \leq \prod_{S \in \mathcal{F} : e \in S} \exp(-K x_S) 
        = \exp\left(-K \sum_{S \in \mathcal{F} : e \in S} x_S\right) \\
        & \leq \exp\left(-K  \sum_{S \in \mathcal{F} : e \in S} x_S\right)
        \leq  \exp(-K )
        \leq \frac{1}{\mathrm{poly}(n)}    
    \end{align*}
    by setting the hidden constant in $K$ to be large enough.
    By a union bound over elements in $U$, the probability that the output is infeasible is at most $1/\mathrm{poly}(n)$.

    It is clear that the expected total weight of $\bm{\mathcal{S}}$ is as stated.

    Now, we consider the Lipschitz continuity.
    For $S \in \mathcal{F}$, let $\delta_S = w'_S - w_S$.
    For $k \in [K]$, let $\bm{\mathcal{S}}_k$ and $\bm{\mathcal{S}}'_k$ be the family of sets added to $\bm{\mathcal{S}}$ and $\bm{\mathcal{S}}'$, respectively, at step $k$.
    We note that
    \begin{align}
        \sum_{S \in \mathcal{F}} |w_S x_S - w'_S x'_S|
        = \sum_{S \in \mathcal{F}} |w_S x_S - (w_S + \delta_S) x'_S|
        \geq \sum_{S \in \mathcal{F}} \left(|x_S - x'_S|w_S - |\delta_S| x'_S\right).
        \label{eq:set-cover-rounding-1}    
    \end{align}
    Then, we have
    \begin{align*}
        & \EMW((\bm{\mathcal{S}},w),(\bm{\mathcal{S}}',w'))
        \leq \sum_{k \in [K]} \EMW((\bm{\mathcal{S}}_k,w),(\bm{\mathcal{S}}'_k,w')) \\
        & \leq \sum_{k \in [K]} \sum_{S \in \mathcal{F}} \left(\min\{x_S,x'_S\} \cdot |w_S - w'_S| + \max\{x_S-x'_S,0\} \cdot w_S + \max\{x'_S - x_S,0\} \cdot w'_S\right) \\
        & = K  \sum_{S \in \mathcal{F}} \Bigl(\min\{x_S,x'_S\} |\delta_S| + \max\{x_S-x'_S,0\}w_S + \max\{x'_S - x_S,0\}(w_S+\delta_S) \Bigr)  \\
        & \leq K \sum_{S \in \mathcal{F}} \Bigl(\max\{x_S,x'_S\} |\delta_S| + |x_S-x'_S| w_S \Bigr) \\
        & \leq K \sum_{S \in \mathcal{F}} \Bigl(\max\{x_S,x'_S\} |\delta_S| + \bigl(|w_S x_S - w'_S x'_S| + x'_S |\delta_S|\bigr) \Bigr)  \tag{by~\eqref{eq:set-cover-rounding-1}} \\
        & \leq K \sum_{S \in \mathcal{F}} \Bigl(2|\delta_S| + |w_S x_S - w'_S x'_S| \Bigr)  \\
        & = O\left(K \|w - w'\|_1 + K \sum_{S \in \mathcal{F}} |w_S x_S - w'_S x'_S| \right).
        \qedhere
    \end{align*}
\end{proof}

\subsection{Putting Things Together}\label{subsec:set-cover-whole}

\begin{algorithm}[t!]
    \caption{Lipschitz continuous algorithm for the set cover problem}\label{alg:set-cover-general}
    \Procedure{\emph{\Call{LPBasedSetCover}{$U,\mathcal{F},w$}}}{
        \KwIn{A set $U$, $\mathcal{F} \subseteq 2^U$ and a weight vector $w \in \mathbb{R}_{\geq 0}^{\mathcal{F}}$.}
        $c \gets $ an arbitrarily small constant\;
        $\bm{\lambda} \gets $ uniformly at random from $[(\OPTLP+c) /\min\{n,m\},2(\OPTLP+c)/\min\{n,m\}]$\;
        $\bm{\kappa} \gets $ uniformly at random from $[(\OPTLP+c)/\log m,2(\OPTLP+c)/\log m]$\;
        $\bm{W} \gets$ uniformly at random from $[3(\OPTLP + 2c),6(\OPTLP+2c)]$\;
        $\bm{x} \gets \Call{RegularizedLP}{U,\mathcal{F},w,\bm{\lambda},\bm{\kappa},\bm{W}}$\;
        $\bm{\mathcal{S}} \gets \Call{RoundingSC}{U,\mathcal{F},\bm{x}}$\;
        \Return $\bm{\mathcal{S}}$.
    }
\end{algorithm}

Our algorithm, \Call{LPBasedSetCover}{}, is given in Algorithm~\ref{alg:set-cover-general}.

\begin{proof}[Proof of Theorem~\ref{thm:set-cover}]
    We first consider the approximation ratio.
    By Lemma~\ref{lem:set-cover-regularized-lp}, $\bm{x}$ satisfies
    \begin{align}
        \sum_{S \in \mathcal{F}} w_S \bm{x}_S \leq \OPTLP + \frac{\bm{\lambda} \min\{n,m\}}{2} + \frac{\bm{\kappa} \log m}{2} 
        \leq 3(\OPTLP+c).
        \label{eq:set-cover-objective}
    \end{align}
    Then by Lemma~\ref{lem:set-cover-rounding}, we have
    \begin{align}
        \E\left[\sum_{S \in \bm{\mathcal{S}}}w_S \right]
        = O(\log n \cdot (\OPTLP+c)) = O(\log n \cdot (\OPT+c)).
        \label{eq:set-cover-rounding-expected-value}
    \end{align}

    Next, we analyze the Lipschitz continuity.
    Fix a weight vector $w \in \mathbb{R}_{\geq 0}^{\mathcal{F}}$ and a set $S \in \mathcal{F}$.
    Let $w'\in \mathbb{R}_{\geq 0}^{\mathcal{F}}$ be a vector obtained from $w$ by increasing $w_S$ by $0 < \delta < c$.
    By Lemma~\ref{lem:seeoneelement}, it suffices to show that the claim holds for these particular $w$ and $w'$.

    We consider the joint distribution $\mathcal{L}$ between $\bm{\lambda}$ and $\bm{\lambda}'$ defined as follows:
    For each $\lambda$, we transport the probability mass for $\bm{\lambda} = \lambda$ to that for $\bm{\lambda}' = \lambda$ as far as possible.
    The remaining mass is transported arbitrarily.
    Then, we have
    \begin{align}
        & \Pr_{(\lambda,\lambda') \sim \mathcal{L}}[\lambda \neq \lambda']
        \leq 1 - \Pr[\lambda' \in [\OPTLP/\min\{n,m\},2\OPTLP/\min\{n,m\}]] \nonumber \\
        & \leq \frac{2(\OPTLP'+c)/\min\{n,m\} - 2(\OPTLP+c)/\min\{n,m\}}{2(\OPTLP'+c)/\min\{n,m\} - (\OPTLP'+c)/\min\{n,m\}} \nonumber \\
        & \leq \frac{2(\OPTLP+c+\delta) - 2(\OPTLP+c)}{2(\OPTLP'+c) - (\OPTLP'+c)} 
        = O\left(\frac{\delta}{\OPTLP'+c}\right).        \label{eq:setcovergeneral_lambda}
    \end{align}
    Let $\mathcal{K}$ (resp., $\mathcal{W}$) be a joint distribution between $\bm{\kappa}$ and $\bm{\kappa}'$ (resp., $\bm{W}$ and $\bm{W}'$) defined as with $\mathcal{L}$.
    Then, a similar calculation shows that
    \begin{align}
        & \Pr_{(\kappa,\kappa') \sim \mathcal{K}}[\kappa \neq \kappa'] 
        = O\left(\frac{\delta}{\OPTLP'+c}\right),\label{eq:setcovergeneral_kappa} \\
        & \Pr_{(W,W') \sim \mathcal{W}}[W \neq W'] 
        = O\left(\frac{\delta}{\OPTLP'+c}\right). \label{eq:setcovergeneral_W}
    \end{align}
    Let $\mathcal{D}$ be the product distribution of $\mathcal{L}$, $\mathcal{K}$, and $\mathcal{W}$.
    By~\eqref{eq:setcovergeneral_lambda},~\eqref{eq:setcovergeneral_kappa}, and~\eqref{eq:setcovergeneral_W} and a union bound, we have
    \begin{align}
         \Pr_{((\lambda,\kappa,W),(\lambda',\kappa',W')) \sim \mathcal{D}}[\lambda \neq \lambda \vee \kappa \neq \kappa' \vee W \neq W' ] = O\left(\frac{\delta}{\OPTLP'+c}\right). \label{eq:setcovergeneral_params}
    \end{align}

    For $\lambda,\kappa,W>0$, let $\bm{\mathcal{S}}_{\lambda,\kappa,W}$ denote the output family of sets $\bm{\mathcal{S}}$ when $\bm{\lambda}=\lambda$, $\bm{\kappa}=\kappa$, and $\bm{W}=W$.
    From~\eqref{eq:set-cover-objective} and the range of $\bm{W}$, we have
    \begin{align*}
        & \sum_{S \in \mathcal{F}} w_S x_S \leq 3(\OPTLP + c)\leq \bm{W},\\
        & \sum_{S \in \mathcal{F}} w_S x'_S \leq \sum_{S \in \mathcal{F}} w'_S x'_S \leq 
        3(\OPTLP'+c) \leq 3(\OPTLP+2c) \leq \bm{W} \\
        & \sum_{S \in \mathcal{F}} w'_S x_S \leq \sum_{S \in \mathcal{F}} w_S x_S + c \leq 3(\OPTLP + c) + c \leq \bm{W},\\
        & \sum_{S \in \mathcal{F}} w'_S x'_S \leq 3(\OPTLP'+c) \leq 3(\OPTLP+2c) \leq \bm{W}.
    \end{align*}
    Hence, the requirement of Lemma~\ref{lem:set-cover-regularized-lp} is satisfied.
    Then, we have
    \begin{align*}
        & \EMW((\bm{\mathcal{S}},w),(\bm{\mathcal{S}}',w')) \\
        & = \E_{((\lambda,\kappa,W),(\lambda',\kappa',W')) \sim \mathcal{D}} \EMW((\bm{\mathcal{S}}_{\lambda,\kappa,W},w),(\bm{\mathcal{S}}'_{\lambda',\kappa',W'},w'))  \\
        & = \E_{((\lambda,\kappa,W),(\lambda',\kappa',W')) \sim \mathcal{D}} \left[\EMW((\bm{\mathcal{S}}_{\lambda,\kappa,W},w),(\bm{\mathcal{S}}'_{\lambda',\kappa',W'},w')) \mid \lambda = \lambda' \wedge \kappa = \kappa' \wedge W = W'\right] \\
        & \qquad \times \Pr_{\mathcal{D}}[\lambda = \lambda \vee \kappa = \kappa' \vee W = W' ] \\
        & \qquad + \Pr_{\mathcal{D}}[\lambda \neq \lambda \vee \kappa \neq \kappa' \vee W \neq W' ] \cdot \log n \cdot  \max\{\OPTLP+c,\OPTLP'+c\} \tag{by~\eqref{eq:set-cover-rounding-expected-value}} \\
        & = \max_{\lambda,\kappa,W} \EMW((\bm{\mathcal{S}}_{\lambda,\kappa,W},w),(\bm{\mathcal{S}}'_{\lambda,\kappa,W},w')) + O\left(\delta \log n\right) \tag{by~\eqref{eq:setcovergeneral_params}} \\
        & = \max_{\lambda,\kappa,W} O\left(\left(\|w-w'\|_1 + \sum_{S \in \mathcal{F}} |w_v x_v - w'_v x'_v|\right)  \log n  \right) + O(\delta \log n) \tag{by Lemma~\ref{lem:set-cover-rounding}} \\        
        & = \max_{\lambda,\kappa,W} O\left(  \delta \left(1 + (\kappa + W)\sqrt{1+\frac{1
         }{\kappa \lambda}} \right) \log n\right)   + O(\delta \log n) \tag{by Lemma~\ref{lem:set-cover-regularized-lp}} \\
        & = O\left(  \delta \left(1 + \left(\frac{\OPTLP+c}{\log m} + (\OPTLP+c)\right)\sqrt{1+\frac{\min\{n,m\} \log m}{(\OPTLP+c)^2}} \right) \log n\right) \\
        & \qquad + O(\delta \log n)  \\
        & = O\left(\delta \sqrt{\min\{n,m\}\log m }\log n\right).
        \qedhere
    \end{align*}
\end{proof}

%% file: set-cover-f.tex

\section{LP-based $(f+\epsilon)$-Approximation Algorithm for Set Cover}\label{sec:set-cover-f}

In this section, we design a Lipschitz-continuous algorithm for the minimum set cover problem with approximation ratio of $f+\epsilon$, where $f$ is the maximum frequency of an element.
Specifically, we show the following:
\begin{theorem}\label{thm:set-cover-freq}
    There exists a polynomial-time randomized algorithm for the minimum set cover problem that, given a set cover instance $(U,\mathcal{F})$ on $n$ elements, each having frequency at most $f \geq 1$, and $m$ sets, $\epsilon > 0$, and a weight vector $w \in \mathbb{R}_{\geq 0}^{\mathcal{F}}$, outputs a family of sets $\bm{\mathcal{S}} \subseteq \mathcal{F}$ with the following properties:
    \begin{itemize}
    \item $\bm{\mathcal{S}}$ is a feasible solution,
    \item the total weight of $\bm{\mathcal{S}}$ satisfies $\E[\sum_{S \in \bm{\mathcal{S}}}w_S] = (f+\epsilon) (\OPT+c)$, where $c$ is an arbitrarily small constant, and
    \item the algorithm has a Lipschitz constant $O(\epsilon^{-1} f^3 \sqrt{\min\{n,m\}\log m} \log n)$.
    \end{itemize}
\end{theorem}

As with the algorithm of Theorem~\ref{thm:set-cover}, we first solve the regularized LP given in Section~\ref{subsec:set-cover-lp}.
However, we use a different rounding scheme that we explain in Section~\ref{subsec:set-cover-rounding-freq}.
We prove Theorem~\ref{thm:set-cover-freq} in Section~\ref{subsec:set-cover-whole-freq}.

\subsection{Rounding}\label{subsec:set-cover-rounding-freq}

\begin{algorithm}[t!]
    \caption{Rounding algorithm for the set cover problem}\label{alg:set-cover-rounding-freq}
    \Procedure{\emph{\Call{RoundingSC-Freq}{$U,\mathcal{F},\epsilon,x$}}}{
        \KwIn{A set $U$, $\mathcal{F} \subseteq 2^U$, $0<\epsilon<1/(2f)$, and a vector $x \in [0,1]^{\mathcal{F}}$.}
        Let $\bm{\tau}$ be uniformly sampled from $[1/f-\epsilon,1/f]$\;
        $\bm{\mathcal{S}} \gets \emptyset$\;
        \For{$S \in \mathcal{F}$}{
            Add $S$ to $\bm{\mathcal{S}}$ if $x_S \geq \bm{\tau}$\;
        }
        \Return $\bm{\mathcal{S}}$.
    }
\end{algorithm}

Our rounding scheme adds all the sets $S \in \mathcal{S}$ with $x_S$ being larger than rougly $1/f$ to the output set (Algorithm~\ref{alg:set-cover-rounding-freq}).
We have the following:
\begin{lemma}\label{lem:set-cover-rounding-freq}
    Let $(U,\mathcal{F})$ be a set cover instance on $n$ elements, each having frequency at most $f>0$, and let $w \in \mathbb{R}_{\geq 0}^{\mathcal{F}}$ be a weight vector.
    Let $\bm{\mathcal{S}} = \Call{RoundingSC-Freq}{U,\mathcal{F},\epsilon,x}$, where $0<\epsilon<1/(2f)$ and $x \in \mathbb{R}^{\mathcal{F}}$ is a feasible solution to LP~\eqref{eq:set-cover-lp}.
    Then, $\bm{\mathcal{S}}$ is a feasible set cover, and
    \[
        \E\left[\sum_{S \in \bm{\mathcal{S}}}w_S\right] = f(1+2\epsilon f) \cdot \sum_{S \in \mathcal{F}}w_S x_S \leq 2f \sum_{S \in \mathcal{F}}w_S x_S.
    \]

    Let $w,w' \in \mathbb{R}_{\geq 0}^{\mathcal{F}}$ be weight vectors and $x,x' \in [0,1]^{\mathcal{F}}$ be feasible solutions to LP~\eqref{eq:set-cover-lp}.
    For $\bm{\mathcal{S}} = \Call{RoundingSC-Freq}{U,\mathcal{F},\epsilon,x}$ and $\bm{\mathcal{S}}' = \Call{RoundingSC-Freq}{U,\mathcal{F},\epsilon,x'}$, we have   \[
        \EMW((\bm{\mathcal{S}},w),(\bm{\mathcal{S}}',w')) = O\left(\frac{1}{\epsilon} \left( \|w-w'\|_1 +\sum_{S \in \mathcal{F}} |w_S x_S  - w'_S x'_S |\right)\right).
    \]
\end{lemma}
\begin{proof}
		The output $\bm{\mathcal{S}}$ is always feasible because $\tau \leq 1/f$ always holds and for any set $e \in U$, there exists a set $S \ni e$ such that $x_S \geq 1/f \geq \tau$.

		It is clear that the expected total weight of $\bm{\mathcal{S}}$ is 
		\[
			\sum_{S \in \bm{\mathcal{S}}} w_S
			\leq 
			\sum_{S \in \mathcal{F}} w_S x_S \cdot \frac{1}{1/f-\epsilon}
			\leq 
			f(1+2\epsilon f)\sum_{S \in \mathcal{F}} w_S x_S,
		\]
		where we used $\epsilon<1/(2f)$ in the last inequality.

    Next, we consider the Lipschitz continuity.
    For a set $S \in \mathcal{F}$, let $\delta_S = w'_S - w_S$.
    Let $p_S = \Pr[S \in \mathcal{S}]$ and $p'_S = \Pr[S \in \mathcal{S}']$.
    Then from the choice of $\bm{\tau}$, we have
    \begin{align}
    	|p_S - p'_S| \leq \frac{|x_S - x'_S|}{\epsilon}.
    	\label{eq:set-cover-rounding-f-1}
    \end{align}
    Recall that (see~\eqref{eq:set-cover-rounding-1})
    \begin{align}
        \sum_{S \in \mathcal{F}} |w_S x_S - w'_S x'_S|
        \geq \sum_{S \in \mathcal{F}} \left(|x_S - x'_S|w_S - |\delta_S| x'_S\right).
        \label{eq:set-cover-rounding-freq}
    \end{align}
    Then, we have
    \begin{align*}
        & \EMW((\bm{\mathcal{S}},w),(\bm{\mathcal{S}}',w')) \\
        & \leq \sum_{S \in \mathcal{F}} \left(\min\{p_S,p'_S\} \cdot |w_S - w'_S| + \max\{p_S-p'_S,0\} \cdot w_S + \max\{p'_S - p_S,0\} \cdot w'_S\right) \\
        & = \sum_{S \in \mathcal{F}} \Bigl(\min\{p_S,p'_S\} |\delta_S| + \max\{p_S-p'_S,0\}w_S + \max\{p'_S - p_S,0\}(w_S+\delta_S) \Bigr)  \\
        & \leq \sum_{S \in \mathcal{F}} \Bigl(\max\{p_S,p'_S\} |\delta_S| + |p_S-p'_S| w_S \Bigr) \\
        & \leq \sum_{S \in \mathcal{F}} \Bigl(\max\{p_S,p'_S\} |\delta_S| + \frac{|x_S-x'_S| w_S}{\epsilon} \Bigr) \tag{by~\eqref{eq:set-cover-rounding-f-1}} \\
        & \leq \sum_{S \in \mathcal{F}} \Bigl(\max\{p_S,p'_S\} |\delta_S| + \frac{|w_S x_S - w'_S x'_S| + x'_S |\delta_S|}{\epsilon} \Bigr)  \tag{by~\eqref{eq:set-cover-rounding-freq}} \\
        & \leq \frac{1}{\epsilon}\sum_{S \in \mathcal{F}} \Bigl(2|\delta_S| + |w_S x_S - w'_S x'_S| \Bigr)  \\
        & = O\left(\frac{1}{\epsilon}\left(\|w - w'\|_1 + \sum_{S \in \mathcal{F}} |w_S x_S - w'_S x'_S| \right)\right).
        \qedhere
    \end{align*}	
\end{proof}

\subsection{Putting Things Together}\label{subsec:set-cover-whole-freq}

\begin{algorithm}[t!]
    \caption{Lipschitz continuous algorithm for the set cover problem}\label{alg:set-cover-general-f}
    \Procedure{\emph{\Call{LPBasedSetCover-Freq}{$U,\mathcal{F},\epsilon, w$}}}{
        \KwIn{A set $U$, $\mathcal{F} \subseteq 2^U$, $0 < \epsilon < 1/(2f)$, and a weight vector $w \in \mathbb{R}_{\geq 0}^{\mathcal{F}}$.}
				$c \gets$ an arbitrarily small constant\;
        $\bm{\lambda} \gets $ uniformly at random from $[\epsilon (\OPTLP+c) /\min\{n,m\},2\epsilon (\OPTLP+c) /\min\{n,m\}]$\;
        $\bm{\kappa} \gets $ uniformly at random from $[\epsilon (\OPTLP+c)/\log m,2\epsilon (\OPTLP+c)/\log m]$\;
        $\bm{W} \gets$ uniformly at random from $[\OPTLP+2c,2(\OPTLP+2c)]$\;
        $\bm{x} \gets \Call{RegularizedLP}{U,\mathcal{F},w,\bm{\lambda},\bm{\kappa},\bm{W}}$\;
        $\bm{\mathcal{S}} \gets \Call{RoundingSC-Freq}{U,\mathcal{F},\epsilon,\bm{x}}$\;
        \Return $\bm{\mathcal{S}}$.
    }
\end{algorithm}

Our algorithm is given in Algorithm~\ref{alg:set-cover-general-f}.
It is similar to Algorithm~\ref{alg:set-cover-general}.
The main difference is that we use smaller $\bm{\lambda}$ and $\bm{\kappa}$ for a better approximation guarantee and that we call \Call{RoundingSC-Freq}{} instead of \Call{RoundingSC}{}.
We prove that this algorithm satisfies the claims of Theorem~\ref{thm:set-cover-freq}:
\begin{proof}[Proof of Theorem~\ref{thm:set-cover-freq}]
    We first consider the approximation ratio.
    By Lemma~\ref{lem:set-cover-regularized-lp}, $\bm{x}$ satisfies
    \[
        \sum_{S \in \mathcal{F}} w_S \bm{x}_S \leq \OPTLP + \frac{\bm{\lambda} \min\{n,m\}}{2} + \frac{\bm{\kappa} \log m}{2} 
        = (1+2\epsilon)(\OPTLP+c).
    \]
    Then by Lemma~\ref{lem:set-cover-rounding-freq}, we have
    \begin{align}
        & \E\left[\sum_{S \in \bm{\mathcal{S}}}w_S \right]
        = f(1+2\epsilon)(1+2\epsilon f)  \cdot (\OPT +c)
        = f(1+O(\epsilon f))  \cdot (\OPT+c).
        \label{eq:set-cover-rounding-expected-value-freq}
    \end{align}

    Next, we analyze the Lipschitz continuity.
    Fix a weight vector $w \in \mathbb{R}_{\geq 0}^{\mathcal{F}}$ and a set $S \in \mathcal{F}$.
    Let $w'\in \mathbb{R}_{\geq 0}^{\mathcal{F}}$ be a vector obtained from $w$ by increasing $w_S$ by $0 < \delta < c$.
    By Lemma~\ref{lem:seeoneelement}, it suffices to show that the claim holds for $w$ and $w'$.

    Following the proof of Theorem~\ref{thm:set-cover}, there exists a joint distribution $\mathcal{D}$ such that 
    \[
         \Pr_{((\lambda,\kappa,W),(\lambda',\kappa',W')) \sim \mathcal{D}}[\lambda \neq \lambda \vee \kappa \neq \kappa' \vee W \neq W' ] = O\left(\frac{\delta}{\OPTLP'+c}\right). \label{eq:setcovergeneral_params-freq}
    \]

    For $\lambda,\kappa,W>0$, let $\bm{\mathcal{S}}_{\lambda,\kappa,W}$ denote the output family of sets $\bm{\mathcal{S}}$ when $\bm{\lambda}=\lambda$, $\bm{\kappa}=\kappa$, and $\bm{W}=W$.
    From~\eqref{eq:set-cover-objective} and the range of $\bm{W}$, we have
    \begin{align*}
        & \sum_{S \in \mathcal{F}} w_S x_S \leq (1+2\epsilon)(\OPTLP + c)\leq \bm{W},\\
        & \sum_{S \in \mathcal{F}} w_S x'_S \leq \sum_{S \in \mathcal{F}} w'_S x'_S \leq 
        (1+2\epsilon)(\OPTLP'+c) \leq (1+2\epsilon)(\OPTLP+2c) \leq \bm{W} \\
        & \sum_{S \in \mathcal{F}} w'_S x_S \leq \sum_{S \in \mathcal{F}} w_S x_S + c \leq (1+2\epsilon)(\OPTLP + c) + c \leq \bm{W},\\
        & \sum_{S \in \mathcal{F}} w'_S x'_S \leq (1+2\epsilon)(\OPTLP'+c) \leq (1+2\epsilon)(\OPTLP+2c) \leq \bm{W}.
    \end{align*}
    Hence, the requirement of Lemma~\ref{lem:set-cover-regularized-lp} is satisfied.
    Then, we have
    \begin{align*}
        & \EMW((\bm{\mathcal{S}},w),(\bm{\mathcal{S}}',w')) \\
        & = \E_{((\lambda,\kappa,W),(\lambda',\kappa',W')) \sim \mathcal{D}} \EMW((\bm{\mathcal{S}}_{\lambda,\kappa,W},w),(\bm{\mathcal{S}}'_{\lambda',\kappa',W'},w'))  \\
        & = \E_{((\lambda,\kappa,W),(\lambda',\kappa',W')) \sim \mathcal{D}} \left[\EMW((\bm{\mathcal{S}}_{\lambda,\kappa,W},w),(\bm{\mathcal{S}}'_{\lambda',\kappa',W'},w')) \mid \lambda = \lambda' \wedge \kappa = \kappa' \wedge W = W'\right] \\
        & \qquad \times \Pr_{\mathcal{D}}[\lambda = \lambda \vee \kappa = \kappa' \vee W = W' ] \\
        & \qquad + \Pr_{\mathcal{D}}[\lambda \neq \lambda \vee \kappa \neq \kappa' \vee W \neq W' ] \cdot f(1+O(\epsilon f)) \cdot  \max\{\OPTLP+c,\OPTLP'+c\} \tag{by~\eqref{eq:set-cover-rounding-expected-value-freq}} \\
        & = \max_{\lambda,\kappa,W} \EMW((\bm{\mathcal{S}}_{\lambda,\kappa,W},w),(\bm{\mathcal{S}}'_{\lambda,\kappa,W},w')) + O\left(\delta f(1+\epsilon f)\right) \tag{by~\eqref{eq:setcovergeneral_params-freq}} \\
        & = \max_{\lambda,\kappa,W} \EMW((\bm{\mathcal{S}}_{\lambda,\kappa,W},w),(\bm{\mathcal{S}}'_{\lambda,\kappa,W},w')) + O(\delta f) \tag{by $\epsilon<1/(2f)$} \\
        & = \max_{\lambda,\kappa,W} O\left(\left(\|w-w'\|_1 + \sum_{S \in \mathcal{F}} |w_v x_v - w'_v x'_v|\right)  f  \right) + O(\delta f) \tag{by Lemma~\ref{lem:set-cover-rounding-freq}} \\        
        & = \max_{\lambda,\kappa,W} O\left(\delta f(\kappa + W)\sqrt{1+\frac{1}{\kappa \lambda}}\right)   + O(\delta f) \tag{by Lemma~\ref{lem:set-cover-regularized-lp}} \\
        & = O\left(\delta f \left(\frac{\epsilon \OPTLP+c}{\log m} + (\OPTLP+c)\right)\sqrt{1+\frac{\min\{n,m\}\log m}{\epsilon^2 (\OPTLP+c)^2}}\right)   + O(\delta f) \\
        & = O\left(\frac{\delta f}{\epsilon} \sqrt{\min\{n,m\}\log m }\log n\right).
        \qedhere
    \end{align*}
    By replacing $\epsilon$ with $O(\epsilon/f^2)$, we obtain the theorem.
\end{proof}    

%% file: feedback-vertex-set.tex

\section{Feedback Vertex Set}\label{sec:fedback-vertex-set}
In this section, we show a Lipschitz continuous algorithm for the minimum feedback vertex set problem.
\begin{theorem}\label{thm:fvs}
    There exists a polynomial-time randomized algorithm for the minimum feedback vertex set that, given a graph $G=(V,E)$ on $n$ vertices and a weight vector $w \in \mathbb{R}_{\geq 0}$, outputs a set of vertices $\bm{S} \subseteq V$ with the following properties:
    \begin{itemize}
    \item $\bm{S}$ is a feedback vertex set with probability $1-1/\mathrm{poly}(n)$,
    \item The total cost of $\bm{S}$ satisfies $\E[\sum_{v \in \bm{S}}w_v] = O(\log n \cdot (\OPT+c))$, where $c>0$ is an arbitrarily small constant, and 
    \item the algorithm has a Lipschitz constant $O(\sqrt{n}\log^2 n)$.
    \end{itemize}
\end{theorem}

As with the set cover problem, our algorithm first solves an LP relaxation and then round the obtained LP solution so that the whole process is Lipschitz continuous.
We use the algorithm in Section~\ref{subsec:set-cover-lp} to solve the LP relaxation.
However, because there are exponentially many constraints in the LP, the rounding algorithm in Section~\ref{subsec:set-cover-rounding} gives a polynomial approximation ratio and Lipschitz constant, which are vacuous.
To address this issue, we round the LP solution using a technique called cycle sparsification, which is explained in Section~\ref{subsec:cycle-sparsification}.
The rounding algorithm is explained in Section~\ref{subsec:fvs-rounding}.
The overall algorithm and its analysis are given in Section~\ref{subsec:fvs-proof}.

\subsection{Cycle Sparsification}\label{subsec:cycle-sparsification}

Let $G=(V,E)$ be a graph and $z \in \mathbb{R}_{\geq 0}^V$ be a weight vector.
For a cycle $C$ in $G$, let $z(C) = \sum_{v\in C}z_v$ denote the total weight of vertices in $C$.
Let $z,\widetilde{z} \in \mathbb{R}_{\geq 0}^V$ be two weight vectors on the same vertex set.
For $\epsilon > 0$, we say that $(G,\widetilde{z})$ is an \emph{$\epsilon$-cycle sparsifier} of $(G,z)$ if for any cycle $C$ in $G$, we have
\[
    (1-\epsilon) z(C) \leq \widetilde{z}(C) \leq (1+\epsilon) z(C)
\]
For a vertex $v \in V$, let $g_v$ denote the \emph{(weighted) girth} of $v$, that is, the minimum weight of a cycle passing through $v$.
We say that a vector $\ell \in \mathbb{R}_{\geq 0}^V$ is a \emph{girth lower bound} if $\ell_v \leq g_v$ for every $v \in V$.

To construct a cycle sparsifier, we consider an independent sampling algorithm using a girth lower bound, described in Algorithm~\ref{alg:cycle-sparsifier}.
In Section~\ref{subsubsec:sparsification-guarantee}, we show that \Call{CycleSparsification}{} (Algorithm~\ref{alg:cycle-sparsifier}) applied on an integer weight vector $z \in \mathbb{Z}_{\geq 0}^V$ outputs a weight vector $\widetilde{\bm{z}} \in \mathbb{R}_{\geq 0}^V$ that sparsifies $z$ with a high probability.
For a vector $z \in \mathbb{R}_{\geq 0}^V$, let $S(z) \subseteq V$ denote the set $\{v \in V : z_v > 0\}$ induced by vertices with positive weights.
In Section~\ref{subsubsec:sparsification-property}, we show several properties of $S(\widetilde{\bm{z}})$ used in subsequent analyses.

\subsubsection{Sparsification Guarantee}\label{subsubsec:sparsification-guarantee}

The goal of this section is to prove the following sparsification guarantee for integer-weighted graphs:
\begin{lemma}\label{lem:cycle-sparsification}
    Let $G=(V,E)$ be a graph on $n$ vertices, $z \in \mathbb{Z}_{\geq 1}^V$ be an integer weight vector, and $\epsilon > 0$.
    Suppose that any cycle has girth at least $\ell \in \mathbb{R}_{\geq 0}^V$.
    Then, \emph{\Call{CycleSparsification}{$G,z,\ell,\epsilon$}} outputs a weight vector $\widetilde{\bm{z}} \in \mathbb{R}_{\geq 0}^V$ such that $(G,\widetilde{\bm{z}})$ is an $\epsilon$-cycle sparsifier of $(G,z)$ with probability $1-1/\mathrm{poly}(n)$.
\end{lemma}

\begin{algorithm}[t]
    \caption{Cycle sparsification for integer-weighted graphs}\label{alg:cycle-sparsifier}
    \Procedure{\emph{\Call{CycleSparsification}{$G,z,\ell,\epsilon$}}}{
        \KwIn{A graph $G=(V,E)$, an integer vector $z \in \mathbb{Z}_{\geq 1}^V$, a girth lower bound $\ell \in \mathbb{R}_{\geq 0}^V$, and $\epsilon > 0$.} 
        $t \gets O(\epsilon^{-2} \log n)$\;
        \For{each $v \in V$}{
            $p_v \gets \min\{t/\ell_v,1\}$\;
            $\widetilde{\bm{z}}_v \gets \mathcal{B}(z_v,p_v) /p_v$, where $\mathcal{B}(n,p)$ denotes the binomial distribution with parameters $n \in \mathbb{Z}_{\geq 1}$ and $p \in [0,1]$.
        }
        \Return $\widetilde{\bm{z}}$.
    }
\end{algorithm}

We prove Lemma~\ref{lem:cycle-sparsification} by following the analysis of the cut sparsification algorithm developed by Fung and Harvey (Theorem 1.15 of~\cite{fung2019general}).
We briefly explain their argument and then discuss modifications we need to prove Lemma~\ref{lem:cycle-sparsification}.

Let $G=(V,E)$ be a graph and $z \in \mathbb{R}_{\geq 0}^E$ be a weight vector.
For a vertex set $S \subseteq V$, let $\delta(S)$ denote the set of edges connecting a vertex in $S$ to $V \setminus S$.
For an edge set $F \subseteq E$, let $z(F) = \sum_{e \in F}z_e$ denote the total weight of edges in $F$.
For two edge weights $z, \widetilde{z} \in \mathbb{R}_{\geq 0}^E$ and $\epsilon > 0$, we say that $(G,\widetilde{z})$ is an \emph{$\epsilon$-cut sparsifier} of $(G,z)$ if for any vertex set $S \subseteq V$, we have 
\[
    (1-\epsilon) z(\delta(S)) 
    \leq \widetilde{z}(\delta(S)) 
    \leq (1+\epsilon) z(\delta(S)).
\]

For an edge $e = (s,t) \in E$, let $k_e$ denote the \emph{(weighted) edge connectivity} of the end-points of $e$, that is, the minimum weight of a cut that separates $s$ and $t$.
To construct a cut sparsifier, Fung~et~al.~\cite{fung2019general} considered an algorithm similar to Algorithm~\ref{alg:cycle-sparsifier}, which loops over edges instead of vertices and uses the edge connectivity $k_e$ instead of the girth $g_v$.
They showed that this algorithm generates an $\epsilon$-cut sparsifier with a high probability. 

The key technical ingredient of their analysis is the cut counting lemma, explained below.
For an edge set $F \subseteq E$, we say that an edge set $F'$ is a \emph{cut-induced subset} of $F$ if $F' = F \cap \delta(S)$ for some $S \subseteq V$.
Then, we can bound the number of small cut-induced subsets as follows.
\begin{lemma}[Theorem 1.6 of~\cite{fung2019general}]\label{lem:cut-counting}
    Let $G=(V,E)$ be a graph on $n$ vertices, $z \in \mathbb{R}_{\geq 0}^E$ be a weight vector, and $F \subseteq E$ be an edge set.
    Suppose that $k_e \geq K$ for every $e \in F$.
    Then, for every real $\alpha \geq 1$, we have
    \[
        |\{\delta(S)\cap F : S \subseteq V \wedge w(\delta(S))\leq \alpha K\}|<n^{2\alpha }.    
    \]
\end{lemma}

To prove Lemma~\ref{lem:cycle-sparsification}, we need a counterpart of Lemma~\ref{lem:cut-counting} for cycles.
For a graph $G$, let $\mathcal{C}(G)$ denote the set of simple cycles in $G$.
For a vertex set $S \subseteq V$, we say that a vertex set $T$ is a \emph{cycle-induced subset} of $S$ if $T = S \cap C$ for some $C \in \mathcal{C}(G)$.
We have the following:
\begin{lemma}\label{lem:cycle-counting}
    Let $G = (V, E)$ be a graph on $m$ edges, $z \in \mathbb{R}_{\geq 0}^V$ be a weight vector, and $S \subseteq V$ be an arbitrary set.
    Suppose that $g_v \geq g$ for every $v \in S$. 
    For any $\alpha \geq 1$, we have
    \[
        |\{C \cap S : C \in \mathcal{C}(G) \wedge z(C)\leq \alpha g\}|<(2m)^{2\alpha+1/2}.
    \]
\end{lemma}
(A slight variation of) this theorem for the case of $S=V$ was proved in~\cite{subramanian1995polynomial}.
Because the proof of Lemma~\ref{lem:cycle-counting} is quite similar, we defer it to Appendix~\ref{apx:cycle-counting}.

Once we obtained Lemma~\ref{lem:cycle-counting}, we can easily prove Lemma~\ref{lem:cycle-sparsification} by following the proof of Fung~et~al.~\cite{fung2019general}.

\subsubsection{Properties of the Induced Set}\label{subsubsec:sparsification-property}

In this section, we show several properties of the set induced by a cycle-sparsifying weight vector.

\begin{lemma}\label{lem:inclusion-probability}
    Let $G=(V,E)$ be a graph, $z \in \mathbb{Z}_{\geq 1}^V$ be an integer weight vector, $\ell \in \mathbb{R}_{\geq 0}^V$ be a girth lower bound, and $\epsilon > 0$.
    Let $\widetilde{\bm{z}} = \Call{CycleSparsification}{G,z,\ell,\epsilon}$.
    Then for every $v \in V$, we have
    \[
        \Pr[v \in S(\widetilde{\bm{z}})] = O\left(\frac{z_v \log n}{\epsilon^2 \ell_v} \right).
    \]
\end{lemma}
\begin{proof}
    Let $q_v$ be the probability that $v \in S(\widetilde{\bm{z}})$.
    Because $z_v \geq 1$, we have $q_v = 1 - (1 - p_v)^{z_v} \leq p_v z_v$, and the claim holds.
\end{proof}

Next, we show a Lipschitz property of \Call{CycleSparsification}{}.
\begin{lemma}\label{lem:cycle-sparsification-lipschitz-unweighted}
    Let $G=(V,E)$ be a graph, $z,z' \in \mathbb{Z}_{\geq 1}^E$ be integer weight vectors, $\ell \in \mathbb{R}_{\geq 0}^V$ be a girth lower bound (with respect to $z$ and $z'$), and $\epsilon > 0$.
    Let $\widetilde{\bm{z}} = \Call{CycleSparsification}{G,z,\ell,\epsilon}$ and $\widetilde{\bm{z}}' = \Call{CycleSparsification}{G',z',\ell,\epsilon}$.
    Then for every $v \in V$, we have 
    \[
        \Bigl| \Pr[v \in S(\widetilde{\bm{z}})] - \Pr[v\in S(\widetilde{\bm{z}}')] \Bigr| = O\left(\frac{|z_v-z'_v|\log n}{\epsilon^2 \ell_v}\right).
    \]
\end{lemma}
\begin{proof}
    Let $q_v$ be the probability that $v \in S(\widetilde{\bm{z}})$, and $\delta_v = z'_v - z_v$.
    When $\delta_v \geq 0$, we have
    \begin{align*}
        & |q_v - q'_v|
        = q_v - q'_v
        = q_v \left(1 - (1 - p_v)^{\delta_v}\right) 
        \leq q_v p_v \delta_v 
        \leq p_v \delta_v,
    \end{align*}
    where the first inequality is by $\delta_v \in \mathbb{R} \setminus (0,1)$.
    A similar calculation shows that $|q_v - q'_v| \leq p_v \delta_v$ when $\delta_v \leq 0$, and hence the claim holds.
\end{proof}

\subsection{Rounding}\label{subsec:fvs-rounding}
In this section, we design a rounding algorithm for the feedback vertex set problem, as given in Algorithm~\ref{alg:fvs-rounding}.
Given a graph $G=(V,E)$ and a weight vector $x \in \mathbb{R}_{\geq 0}^V$, which is supposed to be a solution to LP~\eqref{eq:lp-fvs}, we first compute another weight vector $\bm{z} \in \mathbb{Z}_{\geq 1}^V$ by scaling up $x$ and rounding it up to an integer.
Then, we apply \Call{CycleSparsification}{} to $G$ and $\bm{z}$ to obtain a weight vector $\tilde{\bm{z}} \in \mathbb{R}_{\geq 0}^V$.
Finally, we add every $v \in V$ with $\tilde{\bm{z}}_v > 0$ to the output set $\bm{S}$.

\begin{algorithm}[t]
    \caption{Rounding}\label{alg:fvs-rounding}
    \Procedure{\emph{\Call{RoundingFVS}{$G,x$}}}{
        \KwIn{A graph $G=(V,E)$ and a weight vector $x \in [0,1]^V$.}
        Sample $\bm{b}$ from $[0,1]$ uniformly at random\;
        \For{each $v \in V$}{
            $\bm{z}_v \gets \lceil n^2 x_v + \bm{b}\rceil$.
        }
        $\widetilde{\bm{z}} \gets \Call{CycleSparsification}{G,\bm{z},n^2,1/2}$\;
        $\bm{S} = \{v \in V : \widetilde{\bm{z}}_v > 0\}$\;
        \Return $\bm{S}$.
    }
\end{algorithm}

We first analyze the feasibility of the output and its approximation guarantee.
\begin{lemma}\label{lem:fvs-rounding-approximation}
    Let $G=(V,E)$ be a graph on $n$ vertices and $x \in [0,1]^V$ be a weight vector such that $\sum_{v \in C}x_v \geq 1$ for every cycle $C$ in $G$.
    Let $\bm{S} = \Call{RoundingFVS}{G,x}$.
    Then $\bm{S}$ is a feedback vertex set with probability $1-1/\mathrm{poly}(n)$.
    Moreover, for any vertex $v \in V$, we have 
    \[
        \Pr[v \in \bm{S}] = O\left(x_v \log n + \frac{\log n}{n^2}\right).
    \]
\end{lemma}
\begin{proof}
    First, we check the feasibility of the output.
    Let $C$ be a cycle in $G$.
    Because $C$ has weight at least one in $x$, we have 
    \[
        \sum_{v \in C}\bm{z}_v = \sum_{v \in C} \lceil n^2 x_v + \bm{b}\rceil \geq \sum_{v \in C} n^2 x_v  \geq n^2.
    \]
    Then by Lemma~\ref{lem:cycle-sparsification}, $(G,\widetilde{\bm{z}})$ is an $1/2$-cycle sparsifier of $(G,z)$ with probability $1-1/\mathrm{poly}(n)$.
    When this happens, every cycle $C$ has a positive weight in $\widetilde{\bm{z}}$, and hence the set $\bm{S}$ forms a feedback vertex set.

    Next, by Lemma~\ref{lem:inclusion-probability}, we have
    \[
        \Pr[v \in \bm{S}] 
        = O\left(\E\left[\frac{\bm{z}_v \log n}{n^2}\right]\right)
        = O\left(\E\left[\frac{\lceil n^2 x_v + \bm{b}\rceil \log n}{n^2}\right]\right)
        = O\left(x_v \log n + \frac{\log n}{n^2}\right).
        \qedhere
    \]    
\end{proof}

Next, we consider Lipschitz continuity of \Call{RoundingFVS}{}.
\begin{lemma}\label{lem:fvs-rounding-lipschitz-unweighted}
    Let $G=(V,E)$ be a graph and $x,x' \in [0,1]^V$ be weight vectors with $\sum_{v \in C}x_v \geq 1$ and $\sum_{v \in C}x'_v \geq 1$ for every cycle $C$ in $G$.
    Let $\bm{S} = \Call{RoundFVS}{G,x}$ and $\bm{S}' = \Call{RoundFVS}{G,x'}$.
    Then for any $v \in V$, we have
    \[
        \Bigl|\Pr[v \in \bm{S}] - \Pr[v \in \bm{S}'])\Bigr| = O\left(\frac{|x_v - x'_v|\log n}{\epsilon^2} \right).
    \]
\end{lemma}
\begin{proof}
    To bound the probability difference, we transport the probability mass for $\bm{b}_v = b$ to the mass for $\bm{b}'_v=b$ for $b \in [0,1]$.
    Then by Lemma~\ref{lem:cycle-sparsification-lipschitz-unweighted}, we have 
    \begin{align*}
        & \Bigl| \Pr[v \in \bm{S}] - \Pr[v \in \bm{S}']\Bigr|
        =
        \int_{0}^{1} O\left(\frac{|\lceil n^2 x_v + b\rceil - \lceil n^2 x'_v + b \rceil|\log n}{\epsilon^2 n^2 } \right) \mathrm{d}b \\
        & =
        O\left(\frac{|n^2 (x_v - x'_v)|\log n}{\epsilon^2 n^2 } \right)  
        =
        O\left(\frac{|x_v - x'_v|\log n}{\epsilon^2} \right).
        \qedhere
    \end{align*}
\end{proof}

\subsection{Algorithm and Proof of Theorem~\ref{thm:fvs}}\label{subsec:fvs-proof}

Let $\mathcal{C}$ be the set of all simple cycles in $G$.
The following is a natural LP relaxation for the feedback vertex set problem:
\begin{align}
    \begin{array}{lll}
        \text{minimize} & \displaystyle \sum_{v \in V} w_v x_v, \\
        \text{subject to} & \displaystyle \sum_{v \in C}x_v \geq 1 & \forall C \in \mathcal{C}, \\
        & 0 \leq x_v \leq 1 & \forall v \in V.
    \end{array}
    \label{eq:lp-fvs}
\end{align}
This LP can be seen as the LP relaxation of a set cover instance with exponentially many elements.
However, we can still solve it in polynomial time because there is a polynomial-time separation oracle for the constraints.
Also, it is well known that the integrality gap of this LP is $O(\log n)$~\cite{even2000approximating}.

For a technical reason, we use the following tighter LP relaxation~\cite{chandrasekaran2023polyhedral}:
\begin{align}
    \begin{array}{lll}
        \text{minimize} & \displaystyle \sum_{v \in V} w_v x_v, \\
        \text{subject to} & \displaystyle \sum_{v \in C}x_v \geq 1 & \forall C \in \mathcal{C}, \\
        & x_u + x_v + y_{eu} + y_{ev} \geq 1 & \forall (u,v) \in E\\
        & \displaystyle x_v + \sum_{e \in E: v \in e} y_{ev} \leq 1 & \forall v \in V,  \\
        & 0 \leq x_v \leq 1 & \forall v \in V, \\
        & y_{ev} \geq 0 & \forall e \in E, v \in e.
    \end{array}
    \label{eq:lp-fvs-tight}
\end{align}
We can solve this LP in polynomial time, and it is known that its integrality gap is two~\cite{chandrasekaran2023polyhedral}.

\begin{algorithm}[t]
    \caption{Feedback Vertex Set}\label{alg:fvs}
    \Procedure{\emph{\Call{FeedbackVertexSet}{$G,w,\epsilon$}}}{
        \KwIn{A graph $G=(V,E)$, a weight vector $w \in \mathbb{R}_{\geq 0}^V$, and $\epsilon > 0$.}
        $c \gets an arbitrarily small constant$\;
        $\bm{\lambda} \gets $ uniformly at random from $[(\OPTLP+c)/n,2(\OPTLP+c)/n]$\;
        $\bm{\kappa} \gets $ uniformly at random from $[(\OPTLP+c)/\log n,2(\OPTLP+c)/\log n]$\;
        $\bm{W} \gets$ uniformly at random from $[3(\OPTLP+c),6(\OPTLP+c)]$\;
        \For{$v \in V$}{
            \lIf{$w_v \leq \bm{W}$}{$\bm{w}_v \gets w_v$}
            \lElse{$\bm{w}_v \gets \infty$}
        }
        Solve LP~\eqref{eq:lp-fvs-tight} for the weight vector $\bm{w}$ using the algorithm of Lemma~\ref{lem:set-cover-regularized-lp} with $\bm{\lambda}$, $\bm{\kappa}$, and $\bm{W}$, and let $\bm{x} \in \mathbb{R}^V$ be the obtained solution\;
        $\bm{S} \gets \Call{RoundFVS}{G,\bm{x}}$\;
        \Return $\bm{S}$.
    }
\end{algorithm}

Our algorithm is given in Algorithm~\ref{alg:fvs}.
In this algorithm, we first construct a weight vector $\bm{w} \in \mathbb{R}^V$ by setting $\bm{w}_S = \infty$ when $\bm{w}$ is sufficiently large (This enforces $\bm{x}_S = 0$). 
Then, we apply the algorithm in Section~\ref{subsec:set-cover-lp} to solve LP~\eqref{eq:lp-fvs-tight} for the weight vector $\bm{w}$.
We note that the algorithm can be applied for any minimization problem if the feasible solution is convex.
Let $\bm{x} \in \mathbb{R}_{\geq 0}^V$ be the obtained solution.
Then, we round $\bm{x}$ to an integral solution $\bm{S}$ using \Call{RoundFVS}{}.

We first argue that the optimal LP solution does not degrade a lot by replacing the weight vector $w$ with $\bm{w}$.
\begin{lemma}\label{lem:fvs-lp-modified-value}
    There exists a solution to LP~\eqref{eq:lp-fvs-tight} for the weight vector $\bm{w}$ whose objective value is at most $2\OPTLP$.
\end{lemma}
\begin{proof}
    Let $x^* \in \mathbb{R}^V$ be the optimal solution to LP~\eqref{eq:lp-fvs-tight}.
    Because it has integrality gap of two, there exists a feedback vertex set $S^* \subseteq V$ such that $\sum_{v \in S^*}w_v \leq 2 \sum_{v \in V}w_v x_v = 2\OPTLP$.
    This implies that $S^*$ does not have a vertex $v \in V$ with $w_v > \bm{W} \geq 2\OPTLP$.
    It follows that $\bm{1}_{S^*}$ is a feasible solution to LP~\eqref{eq:lp-fvs-tight} for the weight vector $\bm{w}$, and its objective value is $2\OPTLP$.
\end{proof}

\begin{proof}[Proof of Theorem~\ref{thm:fvs}]
    We first note that $\bm{S}$ is a feasible solution with a high probability by Lemma~\ref{lem:fvs-rounding-approximation}.

    Next, we analyze the solution quality.
    Let $\bm{\OPTLP}$ be the optimal value of LP~\eqref{eq:lp-fvs-tight} for the weight vector $\bm{w}$.
    By Lemma~\ref{lem:fvs-lp-modified-value}, we have $\bm{W} \geq \bm{\OPTLP}$ with probability one.
    Then by Lemma~\ref{lem:set-cover-regularized-lp}, we have
    \[
        \sum_{v \in V}\bm{w}_v \bm{x}_v \leq \OPTLP + \frac{\bm{\lambda}n}{2} + \frac{\bm{\kappa}\log n}{2} = 3(\OPTLP+c).
    \]
    Then by Lemma~\ref{lem:fvs-rounding-approximation}, the expected total weight of $\bm{S}$ is 
    \begin{align*}
        & \sum_{v \in V} \bm{w}_v \Pr[v \in \bm{S}] 
        = 
        O\left(\sum_{v \in V} \bm{w}_v \left( \bm{x}_v \log n + \frac{\log n}{n^2}\right) \right)
        = 
        O\left(\sum_{v \in V} \bm{w}_v \bm{x}_v \log n + \sum_{v \in V} \frac{\bm{W} \log n}{n^2}\right)     \\
        & =
        O\left(\log n \cdot (\OPT+c) + \frac{4n \log n \cdot (\OPTLP+c) }{n^2}\right)
        = O(\log n \cdot (\OPT+c) ).
    \end{align*}

    Next, we consider Lipschitz continuity.
    Let $w,w' \in \mathbb{R}_{\geq 0}^V$ be two weight vectors with $\|w - w'\|_1 \leq \delta < c$.
    Following the proof of Theorem~\ref{thm:set-cover}, there exists a joint distribution $\mathcal{D}$ such that 
    \[
         \Pr_{((\lambda,\kappa,W),(\lambda',\kappa',W')) \sim \mathcal{D}}[\lambda \neq \lambda \vee \kappa \neq \kappa' \vee W \neq W' ] = O\left(\frac{\delta}{\OPTLP'}\right). \label{eq:fvs_params}
    \]
    Then following the proof of Theorem~\ref{thm:set-cover} again, we have
    \begin{align}
        \EMW((\bm{S},w),(\bm{S}',w')) \leq \max_{\lambda,\kappa,W} \EMW((\bm{S}_{\lambda,\kappa,W},w),(\bm{S}'_{\lambda,\kappa,W},w')) + O(\delta \log n).
        \label{eq:fvs-1}
    \end{align}
    Let $\delta_v = |w'_v - w_v|$.
    Then because the rounding process for each vertex is independent, we have
    \begin{align}
        & \EMW((\bm{S}_{\lambda,\kappa,W},w),(\bm{S}'_{\lambda,\kappa,W},w')) \nonumber \\
        & \leq \sum_{v \in V}
        |w_v - w'_v| \cdot \min\{\Pr[v \in \bm{S}_{\lambda,\kappa,W}],\Pr[v \in \bm{S}'_{\lambda,\kappa,W}]\} \nonumber  \\
        & \qquad + \sum_{v \in V} \max\{\Pr[v \in \bm{S}_{\lambda,\kappa,W}] - \Pr[v \in \bm{S}'_{\lambda,\kappa,W}],0\} w_v \nonumber \\
        & \qquad + \sum_{v \in V} \max\{\Pr[v \in \bm{S}'_{\lambda,\kappa,W}] - \Pr[v \in \bm{S}_{\lambda,\kappa,W}],0\} w'_v \nonumber  \\
        & \leq \sum_{v \in V}
        \delta_v \cdot \min\{\Pr[v \in \bm{S}_{\lambda,\kappa,W}],\Pr[v \in \bm{S}'_{\lambda,\kappa,W}]\}   + O\left(\sum_{v \in V} (w_v + w'_v) |x_v - x'_v| \log n\right) \tag{by Lemma~\ref{lem:fvs-rounding-lipschitz-unweighted}} \\
        & \leq \delta + O\left(\sum_{v \in V} (w_v + w'_v) |x_v - x'_v| \log n\right). 
        \label{eq:fvs-2}        
    \end{align}
    We note that 
    \begin{align*}
        & \sum_{v \in V}|w_v x_v - w'_v x'_v| = \sum_{v \in V}|w_v x_v - (w_v + \delta_v) x'_v|
        \geq \sum_{v \in V}\Bigl(|x_v - x'_v|w_v - |\delta_v| x'_v \Bigr) \geq \sum_{v \in V}|x_v - x'_v|w_v - \delta,\\
        & \sum_{v \in V}|w_v x_v - w'_v x'_v| = \sum_{v \in V}|(w_v - \delta_v) x_v - w'_v x'_v|
        \geq \sum_{v \in V}\Bigl(|x_v - x'_v|w'_v - |\delta_v| x_v \Bigr) \geq \sum_{v \in V}|x_v - x'_v|w'_v - \delta. 
    \end{align*}
    Combined with~\eqref{eq:fvs-2}, we have
    \begin{align*}
        & \EMW((\bm{S}_{\lambda,\kappa,W},w),(\bm{S}'_{\lambda,\kappa,W},w'))  \\
        & \leq \delta + O\left(\left(2\sum_{v \in V} |w_v x_v - w'_v x'_v| + 2\delta\right) \log n\right) \\
        & = O\left(\sum_{v \in V} |w_v x_v - w'_v x'_v| \log n + \delta \log n\right) \\
        & =  O\left(\delta \log n \cdot  (\kappa + W)\sqrt{1+\frac{1}{\kappa \lambda}} + \delta \log n\right) \tag{by Lemma~\ref{lem:set-cover-regularized-lp} and $\|\bm{w} - \bm{w}'\|_1 \leq \|w-w'\|_1 = \delta$} \\
        & = O\left(\delta \log n \cdot \sqrt{n \log n} + \delta \log n\right) \\
        & = O\left(\delta \sqrt{n} \log^{3/2}n\right).
    \end{align*}
    Combined with~\eqref{eq:fvs-1}, we have
    \[
        \EMW((\bm{S},w),(\bm{S}',w')) = O\left(\delta \sqrt{n} \log^{3/2}n\right).
    \]
    and the claim holds.
\end{proof}

%% file: shared-randomness.tex
\section{Lipschitz Continuity under Shared Randomness}\label{sec:shared-randomness}

In this section, we show that the Lipschitz continuous algorithms we provided in the previous sections can be made Lipschitz continuous under shared randomness. 
We begin by reviewing the techniques by Kumabe and Yoshida~\cite{kumabe2023lipschitz} to obtain such algorithms.

Let $V$ be a finite set and let \Call{Sample}{} be a sampling process that takes two real values $a,b \in \mathbb{R}$ and a vector $p \in [0,1]^{\mathbb{Z}_{\geq 0}}$ and outputs a real value.
For $c \geq 1$, we say that \Call{Sample}{} is $c$-\emph{stable} for a pair  of functions $(l,r)$ with $l,r: \mathbb{R}_{\geq 0}^{V} \to \mathbb{R}$ if (i) for any $w \in \mathbb{R}_{\geq 0}^{V}$, $\Call{Sample}{l(w),r(w),p}$ is uniformly distributed over $[l(w),r(w)]$ when $p$ follows the uniform distribution over $[0,1]^{\mathbb{Z}_{\geq 0}}$, and (ii) for any $w,w'\in \mathbb{R}_{\geq 0}^{V}$, we have
\begin{align*}
    &\E_{\bm{p}\sim \mathcal{U}\left([0,1]^{\mathbb{Z}_{\geq 0}}\right)}\left[\TV\left(\Call{Sample}{l(w),r(w),\bm{p}},\Call{Sample}{l(w'),r(w'),\bm{p}}\right)\right]\\
    &\leq c\cdot\TV\left(\mathcal{U}([l(w),r(w)]),\mathcal{U}([l(w'),r(w')])\right)
\end{align*}
holds.
They gave the following lemmas.
\begin{lemma}[\rm{\cite{kumabe2023lipschitz}}]\label{lem:shared_const}
Let $l$ and $r$ be constant functions. Then, there is a $1$-stable sampling process for $(l,r)$.
\end{lemma}
\begin{lemma}[\rm{\cite{kumabe2023lipschitz}}]\label{lem:shared_ratio}
Let $c > 1$ and suppose $r(w)=cl(w)$ holds for all $w$. 
Then, there is a $(1+c)$-stable sampling process for $(l,r)$.
\end{lemma}
Now we explain how to modify each of our algorithms to ensure Lipschitz continuity under the shared randomness setting.
\subsection{Algorithm~\ref{alg:vertexcover}}

Algorithm~\ref{alg:vertexcover} can be straightforwardly converted into a Lipschitz continuous algorithm in the shared randomness setting using the process in Lemma~\ref{lem:shared_const} to sample $\bm{z}(e)$, without loss of Lipschitz constant.

\subsection{Algorithm~\ref{alg:simplesetcover_const}}\label{sec:shared_vertexcover}

Algorithm~\ref{alg:simplesetcover_const} can also be straightforwardly converted into a Lipschitz continuous algorithm in the shared randomness setting using the process in Lemma~\ref{lem:shared_const} to sample $\bm{b}$, without loss of Lipschitz constant.

\subsection{Algorithm~\ref{alg:setcover_const}}

Algorithm~\ref{alg:setcover_const} can also be converted into a Lipschitz continuous algorithm in the shared randomness setting without loss of Lipschitz constant as follows. We use the process in Lemma~\ref{lem:shared_const} to sample $\bm{b}$. To sample $\bm{\pi}$, we first sample a parameter $\bm{x}$ from the uniform distribution over $[0,1]$ using the process in Lemma~\ref{lem:shared_const}, and then set $\bm{\pi}=f(\bm{x})$ using a fixed function $f:[0,1] \to \mathcal{P}_{U,s}$, where $\mathcal{P}_{U,s}$ is the set of all permutations over the family $\{A\in U\colon |A|\leq s\}$, such that $\Pr_{\bm{x}\sim \mathcal{U}([0,1])}\left[f(\bm{x})=A\right]$ is equal for all $A\subseteq U$ with $|A|\leq s$.

\subsection{Algorithm~\ref{alg:set-cover-general}}\label{sec:shared_setcover}

The randomness in Algorithm~\ref{alg:set-cover-general} arises from two parts: the sampling of parameter $\bm{\lambda}$, $\bm{\kappa}$, and $\bm{W}$, and \Call{RoundingSC}{}. 
The sampling of the first three parameters can be rewritten using the procedure from Lemma~\ref{lem:shared_ratio} (note that the ranges are dependent on the quantity $\|w\|_1$). 
This modification makes the right-hand side of Equation~\ref{eq:setcovergeneral_lambda} three times larger, but it does not affect the overall analysis.

\Call{RoundingSC}{} can be regarded as a process that determines whether to include each $S\in \mathcal{F}$ in the output. We can regard this algorithm as a procedure to sample $S$ with some probability $q_S$ determined by $x_S$ and $K$. 
It can be implemented as follows: for each $S\in \mathcal{F}$, a value $\bm{x}\in [0,1]$ is uniformly sampled.
If $\bm{x}\leq q_S$, $x$ is included in the solution; otherwise, it is not. 
This operation can be rephrased using Lemma~\ref{lem:shared_const}, without affecting the overall Lipschitz constant.

\subsection{Algorithm~\ref{alg:set-cover-general-f}}\label{sec:shared_setcover_f}

Algorithm~\ref{alg:set-cover-general-f} can be converted into a Lipschitz continuous algorithm in the shared randomness setting in almost the same way as in Section~\ref{sec:shared_setcover}. The only difference is \Call{RoundingSC-Freq}{}. To sample the parameter $\bm{\tau}$ in \Call{RoundingSC-Freq}{}, we use the process in Lemma~\ref{lem:shared_const}.

\subsection{Algorithm~\ref{alg:fvs}}

The randomness in Algorithm~\ref{alg:fvs} arises from three parts: the sampling of parameter $\bm{\lambda}$, $\bm{\kappa}$, and $\bm{W}$, the sampling of parameter $\bm{b}$ in \Call{RoundingFVS}{}, and \Call{CycleSparsification}{}. 
As in Section~\ref{sec:shared_setcover}, the sampling of the first three parameters can be rewritten using the procedure from Lemma~\ref{lem:shared_ratio}. 
As with the sampling of $\bm{b}$ in Section~\ref{sec:shared_vertexcover}, the sampling of $\bm{b}$ can be rewritten using the procedure from Lemma~\ref{lem:shared_const}. 

Let us consider the randomness of Algorithm~\ref{alg:fvs}.
In Algorithm~\ref{alg:fvs}, only whether $\widetilde{\bm{z}}_v > 0$ for each $v\in V$ matters in the output of \Call{CycleSparsification}{}. 
By modifying \Call{CycleSparsification}{} to reassign $\widetilde{\bm{z}}\leftarrow 1$ for all $v\in V$ where $\widetilde{\bm{z}}_v > 0$, the sampling of $\widetilde{\bm{z}}_v$ can be regarded as a procedure to set $\widetilde{\bm{z}}=1$ with some probability $\bm{q}_v$ determined by $\bm{z}_v$ and $p_v$. 
It can be implemented as follows: for each $v\in V$, a value $\bm{x}\in [0,1]$ is uniformly sampled.
If $\bm{x}\leq \bm{q}_v$, $\widetilde{\bm{z}}_v$ is set to be $1$; otherwise, it is set to be $0$. 
This operation can be rephrased using Lemma~\ref{lem:shared_const}, without affecting the overall Lipschitz constant.

%% file: appendix.tex

\section{Naive Bound on Lipschitz Constant}\label{sec:naive-lipschitz-constant-bound}

In this section, we give a proof sketch of the fact that, if there exists an $\alpha$-approximation algorithm $\mathcal{A}$ for a minimization problem on graphs, then for any $\epsilon > 0$, there exists a $(1+\epsilon)\alpha$-approximation algorithm $\widehat{\mathcal{A}}$ with Lipschitz constant $O(\epsilon^{-1}|V|)$.

The algorithm $\widehat{\mathcal{A}}$ works as follows.
Given a graph $G=(V,E)$ and a weight vector $w \in \mathbb{R}_{\geq 0}^V$, we first construct a vector $\widehat{\bm{w}} \in \mathbb{R}_{\geq 0}^V$ by rounding values of $w$ to multiples of $c = \Theta(\epsilon\cdot \OPT/|V|)$ (in a randomized way so that this process has a small Lipschitz constant). 
Then, we apply the algorithm $\mathcal{A}$ on the weighted graph $(G,\widehat{\bm{w}})$.
The approximation ratio is $(1+\epsilon)\alpha$ by choosing the hidden constant in $c$ to be small enough.

Now, we analyze the Lipschitz constant of $\widehat{\mathcal{A}}$.
Let $w,w' \in \mathbb{R}_{\geq 0}^V$ be weight vectors and $\widehat{\bm{w}}$ and $\widehat{\bm{w}}'$ be the vectors obtained by rounding $w$ and $w'$, respectively.
By taking an appropriate coupling $\mathcal{D}$ between $\widehat{\bm{w}}$ and $\widehat{\bm{w}}'$, we have $\E_{\mathcal{D}}\left[\widehat{\bm{w}}\neq \widehat{\bm{w}}'\right]= O(\epsilon^{-1}\|w-w'\|_1\cdot |V|/\OPT)$.
By bounding $d_{\mathrm{w}}((S,w),(S,w'))$ by the trivial bound of $O(\OPT)$, the Lipschitz constant is bounded by
\begin{align*}
& \frac{\EMW((\widehat{\mathcal{A}}(G,w),w),(\widehat{\mathcal{A}}(G,w'),w'))}{\|w-w'\|_1}
\leq \frac{\E_{\mathcal{D}}\left[\widehat{\bm{w}}\neq \widehat{\bm{w}}'\right] \cdot O(\OPT)}{\|w-w'\|_1} \\
& \leq 
\frac{O(\epsilon^{-1}\|w-w'\|_1\cdot |V|/\OPT)\cdot O(\OPT)}{\|w-w'\|_1}=  O\left(\frac{|V|}{\epsilon}\right).
\end{align*}

\section{Proof of Lemma~\ref{lem:cycle-counting}}\label{apx:cycle-counting}

Lemma~\ref{lem:cycle-counting} is an immediate corollary of the following.
\begin{lemma}\label{lem:cycle-counting-half-integral}
    Let $G = (V, E)$ be a graph, $w \in \mathbb{R}_{\geq 0}^V$ be a weight vector, and $S \subseteq V$ be an arbitrary set.
    Suppose that $g_v \geq g$ for every $v \in S$. 
    Then for an integer or a half-integer $k \geq 1$, we have
    \[
        |\{C \cap S \mid C \in \mathcal{C}(G) \wedge w(C)\leq k g\}|<(2m)^{2k},
    \]
    where $m = |E|$.
\end{lemma}
\begin{proof}
    The proof is almost identical to that of the main result of~\cite{subramanian1995polynomial}.

    We say that a cycle $C$ is \emph{light} if $w(C) \leq k g$.
    Pretend that each edge is a road that covers a distance equal to its weight. 
    Imagine that a nondeterministic robot enters the road network at a vertex in $S$ at time $0$ and then goes around a light cycle at a constant speed of one unit. 
    More specifically, imagine that $\sum_{v \in S} d_v$ robots enter the network at time $0$, where $d_v$ is the degree of the vertex $v \in V$ (one robot for each out-going edge of each vertex; thus, there are $d_v$ robots at each vertex at time 0). 
    These robots then move through the network at a speed of one unit. 
    Whenever a robot reaches a vertex of the graph, it splits instantaneously into several copies, leaving the vertex in different directions. 
    The effect of repeated splitting is to increase the number of robots in the network; however, distinct robots traverse distinct trajectories. 
    We require that (i) the trajectory of any given robot be a light cycle and (ii) if a robot reaches a point in the trajectory of another robot and the set of nodes in $S$ that the former robot had passed through is different from that of the latter robot (when it was at that point);
    any robot that cannot meet these conditions immediately destroys itself. 
    A robot that completes its cycle waits there; 
    clearly, each robot completes its cycle by time $k g$ and passes through a different set of nodes in $S$. 
    It follows that the quantity to count equals the number of robots that are alive at time $k g$.

    Let $g^-$ denote a quantity that is infinitesimally smaller than $g$. 
    Let $j$ be any nonnegative integer. 
    We prove the following claim by induction on $j$: there are at most $(2m)^{j+1}$ robots alive at time $jg^-/2$. 
    The case $j$ = 0 is clear because $\sum_{v \in S}d_v \leq 2m$.
    We now assume that the claim is true for $j - 1$ and establish it for $j$.

    The idea is to focus on any one robot alive at time $(j - 1) g^-/2$, and to show that at time $jg^-/2$ this robot has split into at most $2m$ copies. 
    Consider the trajectories of all the above copies during the interval $[(j - 1) g^-/2, jg^-/2]$.
    Once the trajectories of two copies separate and meet again, only one survives, for otherwise, either of the following holds: (i) we have a cycle including a vertex in $S$ with a weight less than $g$, or (ii) one of the two copies destroys itself because the sets of nodes in $S$ in these trajectories are the same at the meeting point.
    We may assume that each robot copy is on some edge of the network at time $jg^-/2$. 
    This is because our robots traverse vertices instantly; if a robot happens to be exactly on a vertex at time $jg^-/2$, we retrace its step by an infinitesimal amount. 
    Now observe that the structure of the trajectories forbids two copies heading the same way on any edge of the network at time $jg^-/2$. 
    It follows that there are at most $2m$ copies, thus establishing the correctness of the inductive step and hence the claim.

    The claim implies that there are at most $(2m)^{2k}$ robots alive at time $(2k - 1) g^-/2$. 
    Let us focus on any one of these robots. 
    We claim that this robot will not split during the interval $[ (2k - 1) g^-/2, kg^-]$.
    This is because at time $kg^-$, the robot must be back at its starting vertex.
    Using the same trajectory argument as before, it follows that at most one robot survives during the interval $[ (2k - 1) g^-/2, kg^-]$. 
    Therefore the total number of successful robots alive at time $kg^-$ is at most $(2m) ^{2k}$. 
    These robots remain at their respective edges until time $kg$.
\end{proof}